\documentclass[12pt]{article}
\usepackage[latin1]{inputenc}
\usepackage{amssymb}
\usepackage{hyperref}
\usepackage{amsmath}
\usepackage{latexsym}
\usepackage{times}
\usepackage{cite}
\usepackage{amssymb}
\usepackage{amsfonts}
\usepackage{amscd}
\usepackage{xcolor}
\usepackage{amstext,amsmath,amssymb,amsfonts}
\newtheorem{theorem}{Theorem}

\newtheorem{corollary}[theorem]{Corollary}

\newtheorem{remark}[theorem]{Remark}

\newtheorem{proof}[theorem]{Proof}

\newcommand{\beq}{\begin{eqnarray}}
\newcommand{\eeq}{\end{eqnarray}}
\newcommand{\beqs}{\begin{eqnarray*}}
	\newcommand{\eeqs}{\end{eqnarray*}}
\newcommand{\bpro}{\begin{pro}}
	\newcommand{\epro}{\end{pro}}
\newcommand{\blem}{\begin{lem}}
	\newcommand{\elem}{\end{lem}}
\newcommand{\bdfn}{\begin{dfn}}
	\newcommand{\edfn}{\end{dfn}}
\newcommand{\bcor}{\begin{cor}}
	\newcommand{\ecor}{\end{cor}}
\newcommand{\bthm}{\begin{thm}}
	\newcommand{\ethm}{\end{thm}}
\newcommand{\bex}{\begin{ex}}
	\newcommand{\eex}{\end{ex}}
\newcommand{\brmk}{\begin{rmk}}
	\newcommand{\ermk}{\end{rmk}}
\newcommand{\bpr}{\begin{pr}}
	\newcommand{\epr}{\end{pr}}
\newcommand{\benum}{\begin{enumerate}}
	\newcommand{\eenum}{\end{enumerate}}
\newcommand{\bitem}{\begin{itemize}}
	\newcommand{\eitem}{\end{itemize}}

\newcommand{\cqfd}{\hfill{\square}}
\chardef\bslash=`\\
\numberwithin{equation}{section}
\numberwithin{table}{section}
\numberwithin{theorem}{section}

\begin{document}
	\begin{center}
		{\Large {Multinomial probability distribution and quantum deformed algebras}}\\
		\vspace{0,5cm}
		 Fridolin Melong\\
		 \vspace{0.25cm}
		 {\em  Institut f\"ur Mathematik, Universit\"at Z\"urich,\\
		 	Winterthurerstrasse 190, CH-8057 Z\"urich, Switzerland}\\
	 	{\em fridomelong@gmail.com}
		\vspace{0.25cm}
	\end{center}
	\today

	\vspace{0.5 cm}
	\begin{abstract}
		The multinomial coefficient and their recurrence relations from the generalized quantum deformed algebras are examined. Moreover, the $\mathcal{R}(p,q)-$ deformed multinomial probability distribution and the negative $\mathcal{R}(p,q)-$ deformed multinomial probability distribution are constructed. The recurrence relations are also determined. Particular cases of our results corresponding to the  quantum algebras in the literature are deduced from the general formalism.
	\end{abstract}
	{\noindent
		{\bf Keywords.}
		$\mathcal{R}(p,q)-$ calculus, quantum algebra, multinomial coefficient,  multinomial distribution, negative multinomial distribution.\\
		MSC (2020)17B37, 81R50, 60E05.	
}
	\tableofcontents
	\section{Introduction}\label{sec1}
	The $q-$ deformation of  Vandermonde formula, Cauchy formula and  univariate discrete probability distributions were investigated in \cite{CA1}.Their properties and limiting distributions were derived .
	Furthermore, the $q-$ deformed multinomial coefficients was defined and their recurrence relations were deduced. Also, the $q-$ deformed multinomial  and negative $q-$ deformed multinomial probability distributions of the first and second kind  were presented \cite{CA2}. 
	
	The $\mathcal{R}(p,q)-$ deformed of
	orthogonal polynomials and basic univariate discrete distributions of the probability theory were defined and discussed by Hounkonnou and Melong\cite{HMD}.
	 Relevant $\mathcal{R}(p,q)-$ deformed factorial moments of a
	random variable  and  associated expressions of mean and variance were presented and established. Besides,
	the recurrence relations for the probability distributions were  derived and 
	 known results in the literature are also recovered as particular cases.
	
	Now, let $p$ and $q$ be two positive real numbers such that $ 0<q<p<1.$ We consider a meromorphic function $\mathcal{R}$ defined on $\mathbb{C}\times\mathbb{C}$ by\cite{HB}:\begin{eqnarray}\label{r10}
	\mathcal{R}(u,v)= \sum_{s,t=-l}^{\infty}r_{st}u^sv^t,
	\end{eqnarray}
	with an eventual isolated singularity at the zero, 
	where $r_{st}$ are complex numbers, $l\in\mathbb{N}\cup\left\lbrace 0\right\rbrace,$ $\mathcal{R}(p^n,q^n)>0,  \forall n\in\mathbb{N},$ and $\mathcal{R}(1,1)=0$ by definition. We denote by $\mathbb{D}_{R}$ the bidisk \begin{eqnarray*}
		\mathbb{D}_{R}&:=&\prod_{j=1}^{2}\mathbb{D}_{R_j}\nonumber\\
		&=&\left\lbrace e=(e_1,e_2)\in\mathbb{C}^2: |e_j|<R_{j} \right\rbrace,
	\end{eqnarray*}
	where $R$ is the convergence radius of the series (\ref{r10}) defined by Hadamard formula\cite{TN}:
	\begin{eqnarray*}
		\lim\sup_{s+t \longrightarrow \infty} \sqrt[s+t]{|r_{st}|R^s_1\,R^t_2}=1.
	\end{eqnarray*}
	We denote by 
	${\mathcal O}(\mathbb{D}_R)$ the set of holomorphic functions defined
	on $\mathbb{D}_R.$

	The  $\mathcal{R}(p,q)-$ deformed numbers is given by \cite{HB}: 
	\begin{eqnarray}\label{rpqnumber}
		[n]_{\mathcal{R}(p,q)}:= \mathcal{R}(p^n,q^n),\qquad n\in\mathbb{N}
	\end{eqnarray} 
	leading to define $\mathcal{R}(p,q)-$ deformed factorials as:
	\begin{eqnarray*}
		[n]!_{\mathcal{R}(p,q)}:= \left\{\begin{array}{lr} 1 \quad \mbox{for} \quad n=0 \quad \\
			\mathcal{R}(p,q)\cdots \mathcal{R}(p^n,q^n) \quad \mbox{for} \quad n\geq 1, \quad \end{array} \right.
	\end{eqnarray*}
	and the  $\mathcal{R}(p,q)-$ deformed binomial coefficients:
	\begin{eqnarray*}
		\left[\begin{array}{c} m  \\ n\end{array} \right]_{\mathcal{R}(p,q)}:=
		\frac{[m]!_{\mathcal{R}(p,q)}}{[n]!_{\mathcal{R}(p,q)}},\quad (m, n)\in\mathbb{N}\cup\{0\};\quad m\geq n.
	\end{eqnarray*}

	The linear operators   on ${\mathcal O}(\mathbb{D}_R)$ are defined by:
	\begin{eqnarray*}\label{operat}
		&&\quad Q: \varphi \longmapsto Q\varphi(z) := \varphi(qz)
		\nonumber \\
		&&\quad P: \varphi \longmapsto P\varphi(z): = \varphi(pz)
		,
	\end{eqnarray*}
	and the $\mathcal{R}(p,q)-$ deformed derivative given as follows:
	\begin{eqnarray*}{\label{deriva1}}
	\partial_{{\mathcal R},p,q} := \partial_{p,q}\frac{p - q}{P-Q}{\mathcal R}(P, Q)
	= \frac{p - q}{pP-qQ}{\mathcal R}(pP, qQ)\partial_{p,q},
	\end{eqnarray*}
	where $\partial_{p,q}$ is the $(p,q)-$ derivative:
	\begin{eqnarray*}
		\partial_{p,q}:\varphi \longmapsto
		\partial_{p,q}\varphi(z) := \frac{\varphi(pz) - \varphi(qz)}{z(p-q)}.
	\end{eqnarray*}
	 We  introduced the quantum algebra associated with the $\mathcal{R}(p,q)-$ deformation. It
	is a quantum algebra, ${\mathcal A}_{\mathcal{R}(p,q)}$, generated by the
	set of operators $\{1, A, A^{\dagger}, N\}$ satisfying the following
	commutation relations \cite{HB1}:
	\begin{eqnarray}
	&& \label{algN1}
	\quad A A^\dag= [N+1]_{\mathcal{R}(p,q)},\quad\quad\quad A^\dag  A = [N]_{\mathcal{R}(p,q)}.
	\cr&&\left[N,\; A\right] = - A, \qquad\qquad\quad \left[N,\;A^\dag\right] = A^\dag
	\end{eqnarray}
	with its realization on ${\mathcal O}(\mathbb{D}_R)$ given by:
	\begin{eqnarray*}\label{algNa}
		A^{\dagger} := z,\qquad A:=\partial_{\mathcal{R}(p,q)}, \qquad N:= z\partial_z,
	\end{eqnarray*}
	where $\partial_z:=\frac{\partial}{\partial z}$ is the usual derivative on $\mathbb{C}.$
	Let us recall some notions useful in this paper.
	
	The model deformation structure functions  $\tau_i, i\in\{1,2\},$ depending on the deformation parameters $p$ and $q$ was introduced in \cite{HMD}.  
	
	For $a,b\in\mathbb{N},$ the  $\mathcal{R}(p,q)-$ deformed shifted factorial is defined by \cite{HMD}:
	\begin{equation*}\label{a}
	\big(a \oplus b\big)^n_{\mathcal{R}(p,q)}: = \displaystyle \prod_{i=1}^{n}\big(a\,\tau^{i-1}_1 + b\,\tau^{i-1}_2\big),\quad \mbox{with}\quad  \big(a \oplus b\big)^0_{\mathcal{R}(p,q)}: =1.
	\end{equation*}
	{Analogously, 
		\begin{eqnarray*}
		\big(a \ominus b\big)^n_{\mathcal{R}(p,q)}: = \displaystyle \prod_{i=1}^{n}\big(a\,\tau^{i-1}_1 - b\,\tau^{i-1}_2\big),\quad \mbox{with}\quad  \big(a \ominus b\big)^0_{\mathcal{R}(p,q)}: =1.
		\end{eqnarray*}
		Furthermore, the $\mathcal{R}(p,q)-$ deformed factorial of $a$ of order $r$ is defined by\cite{HMRC}:
		\begin{eqnarray}
		[a]_{r,\mathcal{R
		}(p,q)}=\prod^{r}_{i=1}[a-i+1]_{\mathcal{R
		}(p,q)},\quad r\in\mathbb{N},
		\end{eqnarray}
		and the following relations hold :
		\begin{eqnarray}\label{011}
		[a]_{\mathcal{R}(p^{-1},q^{-1})} &=& (\tau_1\,\tau_2)^{1-a}\,[a]_{\mathcal{R}(p,q)},\end{eqnarray}
		\begin{eqnarray}\label{014}
		[a]_{\mathcal{R}(p^{-1},q^{-1})}!= (\tau_1\tau_2)^{-{r  \choose 2}}\,[a]_{\mathcal{R}(p,q)}!,
		\end{eqnarray}
		and
		\begin{eqnarray}
		[a]_{r,\mathcal{R}(p^{-1},q^{-1})}
		= (\tau_1\,\tau_2)^{{-ar + {r +1 \choose 2}}}\,[a]_{r,\mathcal{R}(p,q)}\label{015}.
		\end{eqnarray}
		
		Our aims is to construct the multinomial coeficient,  the multinomial probability distribution and properties corresponding to the  $\mathcal{R}(p,q)-$ deformed quantum algebras \cite{HB1}. 
		
		This paper is organized as follows: section $2$ is reserved to  multinomial coefficient associated to the $\mathcal{R}(p,q)-$ deformed quantum algebras. Another versions and their recurrence relations are derived.  In section $3,$ we contruct the  $\mathcal{R}(p,q)-$ deformed multinomial probability distribution of the first and second kind.  Section $4$ is dedicated to  particular cases of our results corresponding to the quantum algbras known in the literature are derived from the general formalism.
	\section{$\mathcal{R}(p,q)-$ deformed multinomial formulae}
This section is reserved to the investigation of the  multinomial coefficients, multinomial formula and negative multinomial formula in the framework of the $\mathcal{R}(p,q)-$ deformed quantum algebras.
 Their recurrence relations are also determined.
\begin{theorem}
	\begin{small}
		The $\mathcal{R}(p,q)-$ deformed multinomial coefficient :
		\begin{eqnarray}\label{eq2.1}
		\bigg[\begin{array}{c}
		x\\ r_1,r_2,\cdots,r_k
		\end{array}\bigg]_{\mathcal{R}(p,q)}:=\frac{[x]_{r_1+r_2+\cdots+r_k,\mathcal{R}(p,q)}}{ [r_1]_{\mathcal{R}(p,q)}![r_2]_{\mathcal{R}(p,q)}!\cdots[r_k]_{\mathcal{R}(p,q)}!}
		\end{eqnarray}
		satisfies the recurrence relation:
		\begin{eqnarray}\label{eq2.3}
		&&\bigg[\begin{array}{c}
		x\\ r_1,\cdots,r_k
		\end{array}\bigg]_{\mathcal {R}(p,q)}=\tau^{s_k}_1\bigg[\begin{array}{c}
		x-1\\ r_1,\cdots,r_k
		\end{array}\bigg]_{\mathcal {R}(p,q)}
		+\tau_2^{x-m_1}\bigg[\begin{array}{c}
		x-1\\ r_1-1,r_2,\cdots,r_k
		\end{array}\bigg]_{\mathcal {R}(p,q)}\cr&&\qquad\quad+
		\tau_2^{x-m_2}\bigg[\begin{array}{c}
		x-1\\ r_1,r_2-1,\cdots,r_k
		\end{array}\bigg]_{\mathcal{R}(p,q)}
		+\cdots+\tau_2^{x-m_k}\bigg[\begin{array}{c}
		x\\ r_1,r_2,\cdots,r_{k-1}
		\end{array}\bigg]_{\mathcal {R}(p,q)}.
		\end{eqnarray}
		Or, equivalently,
		\begin{eqnarray}\label{eq2.4}
		&&\bigg[\begin{array}{c}
		x\\r_1,\cdots,r_k
		\end{array}\bigg]_{\mathcal {R}(p,q)}=\tau_2^{s_k}\bigg[\begin{array}{c}
		x-1\\r_1,\cdots,r_k
		\end{array}\bigg]_{\mathcal {R}(p,q)}
		+\tau^{x-m_1}_1\bigg[\begin{array}{c}
		x-1\\r_1-1,\cdots,r_k
		\end{array}\bigg]_{\mathcal {R}(p,q)}\cr&&+\tau^{x-m_2}_1\tau_2^{s_1}\bigg[\begin{array}{c}
		x-1\\r_1,r_2-1\cdots,r_k
		\end{array}\bigg]_{\mathcal{R}(p,q)}
		+\cdots+\tau^{x-m_k}_1\tau_2^{s_{k-1}}\bigg[\begin{array}{c}
		x-1\\r_1,\cdots,r_{k-1}
		\end{array}\bigg]_{\mathcal{R}(p,q)},
		\end{eqnarray}
		where $r_j\in\mathbb{N}$ and $j\in\{1,2,\cdots,k\},$ with $m_j=\sum_{i=j}^kr_i$ and $s_j=\sum_{i=1}^jr_i.$
	\end{small}
\end{theorem}
\begin{proof}
	Since $$[x]_{s_k, \mathcal{R}(p,q)}=[x]_{\mathcal{R}(p,q)}\,[x-1]_{{s_k-1}, \mathcal{R}(p,q)},$$ $$[x-1]_{s_k, \mathcal{R}(p,q)}=[x-1]_{{s_k-1}, \mathcal{R}(p,q)}\,[x-s_k]_{ \mathcal{R}(p,q)}$$ and $$[x]_{ \mathcal{R}(p,q)}=\tau^{s_k}_1\,[x-s_k]_{\mathcal{R}(p,q)} + \tau^{x-s_k}_2\,[s_k]_{{\mathcal R(p,q)}}.$$ Then, the $\mathcal{R}(p,q)-$ deformed factorials of $x$ of order $s_k=\sum_{i=1}^{k}\,r_k$ satisfies de recursion relation:
	\begin{eqnarray}\label{ath1}
	[x]_{s_k,\mathcal{R}(p,q)}=\tau^{s_k}_1\,[x-1]_{s_k, \mathcal {R}(p,q)} + \sum_{j=1}^{k}\tau^{x-m_j}_2\,[r_j]_{ \mathcal{R}(p,q)}\,[x-1]_{{s_k-1}, \mathcal{R}(p,q)}.
	\end{eqnarray} 
	Multiplying both members of the relation (\ref{ath1}) by $1/[r_1]_{\mathcal{R}(p,q)}![r_2]_{\mathcal{R}(p,q)}!\cdots[r_k]_{\mathcal{R}(p,q)}!$ and using the ${\mathcal R}(p,q)-$ deformed multinomial coefficient (\ref{eq2.1}), we obtain the relation (\ref{eq2.3}). Similarly, the $\mathcal {R}(p,q)-$ deformed number can be expressed as follows:
	\begin{eqnarray*}
		[x]_{\mathcal {R}(p,q)}
		=\tau^{s_k}_2\,[x-s_k]_{\mathcal{R}(p,q)} + \sum_{j=1}^{k} \tau^{x-m_j}_1\,\tau^{s_{j-1}}_2\,[r_j]_{\mathcal{R}(p,q)}
	\end{eqnarray*}
	and the $\mathcal {R}(p,q)-$ deformed factorial of $x$ of order $s_k$ satisfies the recursion relation
	\begin{eqnarray}\label{eq2.6}
	[x]_{s_k,\mathcal{R}(p,q)}=\tau^{s_k}_2\,[x-1]_{s_k, \mathcal{R}(p,q)} + \sum_{j=1}^{k}\tau^{x-m_j}_1\,\tau^{s_{j-1}}_2\,[r_j]_{ \mathcal{R}(p,q)}\,[x-1]_{{s_k-1}, \mathcal{R}(p,q)},
	\end{eqnarray} 
	with $s_0=0.$ 
	
	Dividing the both members of the relation (\ref{eq2.6}) by $[r_1]_{\mathcal{R}(p,q)}![r_2]_{\mathcal{R}(p,q)}!\cdots[r_k]_{\mathcal{R}(p,q)}!$ and using (\ref{eq2.1}), the relation (\ref{eq2.4}) is readily derived and the proof is achieved.  
	$\cqfd$
\end{proof}
\begin{remark}
	\begin{enumerate}
		\item [(i)] From the relations (\ref{011}), (\ref{014}) and (\ref{015}), we obtain the $\mathcal{R}(p^{-1},q^{-1})$- deformed multinomial coefficient in the simpler form:
		\begin{small}
			\begin{eqnarray}\label{eq2.2}
			\bigg[\begin{array}{c}
			x\\r_1,r_2,\cdots,r_k\end{array}\bigg]_{\mathcal{R}(p^{-1},q^{-1})}
			=(\tau_1\tau_2)^{-\displaystyle\sum_{j=1}^k r_j(x-m_j)}\bigg[\begin{array}{c}
			x\\r_1,r_2,\cdots,r_k\end{array}\bigg]_{\mathcal{R}(p,q)}\end{eqnarray}
			or
			\begin{eqnarray}
			\bigg[\begin{array}{c}
			x\\r_1,r_2,\cdots,r_k\end{array}\bigg]_{\mathcal{R}(p^{-1},q^{-1})}=(\tau_1\tau_2)^{-\displaystyle\sum_{j=1}^k r_j(x-s_j)}\bigg[\begin{array}{c}
			x\\r_1,r_2,\cdots,r_k\end{array}\bigg]_{\mathcal{R}(p,q)},
			\end{eqnarray}
			where $s_j=\displaystyle\sum_{i=1}^jr_i,$  $m_j=\displaystyle\sum_{i=j}^kr_i,$ $r_j\in\mathbb{N},$ $j\in\{1,2,\cdots,k\}$ and $k\in\mathbb{N}.$ 
		\end{small}
		\item[(ii)]	Another  recurrence relation can be obtained by replacing $\mathcal{R}(p,q)$ by ${\mathcal R}(p^{-1},q^{-1}),$  and using the expression  (\ref{eq2.2}), respectively. Thus,  the  recursion relations (\ref{eq2.3}) and (\ref{eq2.4}) take the following forms:
		\begin{small}
			\begin{eqnarray} \label{eq2.5}
			&&\bigg[\begin{array}{c}
			x\\r_1,\cdots,r_k\end{array}\bigg]_{\mathcal {R}(p,q)}=\tau_2^{m_1}\bigg[\begin{array}{c}
			x-1\\r_1,\cdots,r_k\end{array}\bigg]_{\mathcal{R}(p,q)}
			+\tau_2^{m_2}\bigg[\begin{array}{c}
			x-1\\r_1-1,r_2,\cdots,r_k\end{array}\bigg]_{\mathcal {R}(p,q)}\cr&&
			\qquad\qquad\qquad+\tau_2^{m_3}\bigg[\begin{array}{c}
			x-1\\r_1,r_2-1,\cdots,r_k\end{array}\bigg]_{\mathcal{R}(p,q)}
			+\cdots+\tau^x_1\bigg[\begin{array}{c}
			x-1\\r_1,r_2,\cdots,r_{k-1}\end{array}\bigg]_{\mathcal{R}(p,q)}
			\end{eqnarray}
			and
			\begin{eqnarray}\label{eq2.5a}
			&&\bigg[\begin{array}{c}
			x\\r_1,\cdots,r_k\end{array}\bigg]_{{\mathcal R}(p,q)}=\tau^x_1\bigg[\begin{array}{c}
			x-1\\r_1,\cdots,r_k\end{array}\bigg]_{\mathcal{R}(p,q)}
			+\tau_2^{x-s_1}\bigg[\begin{array}{c}
			x-1\\r_1-1,r_2,\cdots,r_k\end{array}\bigg]_{\mathcal{R}(p,q)}\cr&&\qquad\qquad\qquad+\tau_2^{x-s_2}\bigg[\begin{array}{c}
			x-1\\r_1,r_2-1,\cdots,r_k\end{array}\bigg]_{\mathcal{R}(p,q)}
			+\cdots+\tau_2^{x-s_k}\bigg[\begin{array}{c}
			x-1\\r_1,\cdots,r_{k-1}\end{array}\bigg]_{\mathcal{R}(p,q)}.
			\end{eqnarray}
			\end{small}
			\item[(iii)] The $q-$ multinomial coefficients and formula given in \cite{CA2} can be recovered by taking $\mathcal{R}(x,1)=(1-q)^{-1}(1-x)$ involving $\tau_1=1$ and $\tau_2=q.$
			\item[(iv)] Taking $k=1,$ we obtained the $\mathcal{R}(p,q)-$ deformed binomial coefficients and related relations \cite{HMD,HMRC}.
	\end{enumerate} 
\end{remark}

Let us  generalized the multinomial formulas in the framework of the $\mathcal {R}(p,q)-$ deformed quantum algebras.
\begin{theorem}
	For $n$ a positive integers, $x,p,$ and $q$ real numbers, the following relation holds:
	\begin{small}
		\begin{eqnarray}
		\prod_{j=1}^{k}\big(1\oplus x_j\big)^{n}_{\mathcal{R}(p,q)}&=&\sum\bigg[\begin{array}{c}
		n\\r_1,\cdots,r_k
		\end{array}\bigg]_{\mathcal {R}(p,q)}\prod_{j=1}^{k}x^{r_j}_j\tau^{n-r_j\choose 2}_1\tau^{r_j\choose 2}_2\nonumber\\&\times&\Big(\tau^{n-s_{j-1}}_1\oplus x_j\tau^{n-s_{j-1}}_2\Big)^{s_{j-1}}_{\mathcal{R}(p,q)},
		\end{eqnarray}
		where  $r_j\in\{0,\cdots,n\},$ $j\in\{1,\cdots,k\},$  with $\sum_{i=1}^{k}r_i\leq n$ and $s_j=\sum_{i=1}^{j}r_i,$ $s_0=0.$
	\end{small}
\end{theorem}
\begin{proof} Setting
	\begin{eqnarray*}
		s_n(x_1,\cdots,x_k;p,q)&=&\sum\,\bigg[\begin{array}{c}
			n\\r_1,\cdots,r_k
		\end{array}\bigg]_{\mathcal {R}(p,q)}\,\prod_{j=1}^{k}x^{r_j}_j\tau^{n-r_j\choose 2}_1\tau^{r_j\choose 2}_2\nonumber\\&\times&\Big(\tau^{n-s_{j-1}}_1\oplus x_j\tau^{n-s_{j-1}}_2\Big)^{s_{j-1}}_{\mathcal{R}(p,q)}
	\end{eqnarray*}
	and using
	\begin{eqnarray}\label{id}
	\bigg[\begin{array}{c}
	n\\r_1
	\end{array}\bigg]_{{\mathcal R}(p,q)}\bigg[\begin{array}{c}
	n-s_1\\r_2
	\end{array}\bigg]_{{\mathcal R}(p,q)}\cdots\bigg[\begin{array}{c}
	n-s_{k-1}\\r_k
	\end{array}\bigg]_{{\mathcal R}(p,q)}=\bigg[\begin{array}{c}
	n\\r_1,\cdots,r_k
	\end{array}\bigg]_{{\mathcal R}(p,q)},
	\end{eqnarray}
	we get:
	\begin{eqnarray*}
		s_n(x_1,\cdots,x_k;p,q)&=&\prod_{j=1}^{k}\bigg(\sum_{r_j=0}^{n-s_{j-1}}\bigg[\begin{array}{c}
			n-s_{j-1}\\r_j
		\end{array}\bigg]_{\mathcal {R}(p,q)}x^{r_j}_j\tau^{n-r_j\choose 2}_1\tau^{r_j\choose 2}_2\bigg)\nonumber\\&\times&\Big(\tau^{n-s_{j-1}}_1\oplus x_j\tau^{n-s_{j-1}}_2\Big)^{s_{j-1}}_{\mathcal{R}(p,q)}.
	\end{eqnarray*}
	From the ${\mathcal R}(p,q)-$ deformed binomial formula, the $j^{th}-$ sum is 
	\begin{eqnarray*}
		\Big(1 \oplus x_j\Big)^{n-s_{j-1}}_{\mathcal{R}(p,q)}=\sum_{r_j=0}^{n-s_{j-1}}\,\bigg[\begin{array}{c}
			n-s_{j-1}\\r_j
		\end{array}\bigg]_{\mathcal {R}(p,q)}\,x^{r_j}_j\tau^{n-r_j\choose 2}_1\tau^{r_j\choose 2}_2,
	\end{eqnarray*}
	where $j\in\{1,2,\cdots,k\}.$ Moreover, 
		\begin{eqnarray*}
			\Big(1 \oplus x_j\Big)^{n-s_{j-1}}_{\mathcal{R}(p,q)}\Big(\tau^{n-s_{j-1}}_1\oplus x_j\,\tau^{n-s_{j-1}}_2\Big)^{s_{j-1}}_{\mathcal{R}(p,q)}
			= \Big(1 \oplus x_j\Big)^{n}_{\mathcal{R}(p,q)},
		\end{eqnarray*}
	with $j\in\{1,2,\cdots,k\}.$ Thus
	\begin{eqnarray*}
		s_n(x_1,\cdots,x_k;p,q)=\prod_{j=1}^{k}\big(1\oplus x_j\big)^n_{\mathcal{R}(p,q)}.
	\end{eqnarray*}
		$\cqfd$
\end{proof}
\begin{theorem}
	Let $n$ be a positive integers, $p$ and $q$ real numbers. Then,
	\begin{small}
		\begin{eqnarray*}
			\prod_{j=1}^{k}\big(1\oplus x_j\big)^{n}_{\mathcal {R}(p,q)}=\sum\bigg[\begin{array}{c}n+s_{k}-1\\r_1,r_2,\cdots,r_k\end{array}\bigg]_{\mathcal {R}(p,q)}\prod_{j=1}^{k}\frac{x^{r_j}_j\tau^{n-r_j\choose 2}_1\tau^{r_j\choose 2}_2}{ \Big(\tau^{n}_1\oplus x_j\tau^{n}_2\Big)^{s_k-s_{j-1}}_{\mathcal{R}(p,q)}}.
		\end{eqnarray*}
		Equivalently,
		\begin{eqnarray*}
			\prod_{j=1}^{k}\big(1\oplus x_j\big)^{n}_{\mathcal {R}(p,q)}=\sum_{r_j\in\mathbb{N}}\bigg[\begin{array}{c}n+s_{k}-1\\r_1,\cdots,r_k\end{array}\bigg]_{\mathcal {R}(p,q)}\prod_{j=1}^{k}\frac{x^{n+s_k-s_{j-1}}_j\tau^{n-r_j\choose 2}_1\tau^{{n+s_k-s_{j-1}\choose 2}+r_j}_2}{\Big(\tau^{n}_1\oplus x_j\tau^{n}_2\Big)^{s_k-s_{j-1}}_{\mathcal{R}(p,q)}},
		\end{eqnarray*} 
		where $j\in\{1,2,\cdots,k\},$ with $s_j=\displaystyle\sum_{i=1}^{j}r_j,\quad s_0=0.$
	\end{small}
\end{theorem}
\begin{proof} Consider the multiple sum defined as follows:
	\begin{small}
		\begin{eqnarray*}
			s_n(x_1,\cdots,x_k;p,q)=\sum_{r_j=0}^{\infty}\,\bigg[\begin{array}{c}
				n+s_k-1\\r_1,r_2,\cdots,r_k
			\end{array}\bigg]_{\mathcal {R}(p,q)}\prod_{j=1}^{k}\frac{x^{r_j}_j\tau^{n-r_j\choose 2}_1\tau^{r_j\choose 2}_2}{\Big(\tau^{n}_1\oplus x_j\tau^{n}_2\Big)^{s_k-s_{j-1}}_{\mathcal{R}(p,q)}}
		\end{eqnarray*}
		and using the relation (\ref{id}), with $n+s_k-1$ instead of $n,$ we obtain:
		\begin{eqnarray*}
			s_n(x_1,\cdots,x_k;p,q)=\prod_{j=1}^{k}\bigg(\sum_{r_j=0}^{\infty}\,\bigg[\begin{array}{c}
				n-s_k-s_{j-1}\\r_j
			\end{array}\bigg]_{\mathcal{R}(p,q)}\frac{x^{r_j}_j\tau^{n-r_j\choose 2}_1\tau^{r_j\choose 2}_2}{ \Big(\tau^{n}_1\oplus x_j\tau^{n}_2\Big)^{s_k-s_{j-1}}_{\mathcal{R}(p,q)}}\bigg).
		\end{eqnarray*}
		From the negative ${\mathcal R}(p,q)-$ deformed binomial formula \cite{HMRC}, we get: 
		\begin{eqnarray}
		\sum_{r_j=0}^{\infty}\,\bigg[\begin{array}{c}
		n-s_k-s_{j-1}\\r_j
		\end{array}\bigg]_{\mathcal {R}(p,q)}\frac{x^{r_j}_j\tau^{n-r_j\choose 2}_1\tau^{r_j\choose 2}_2}{ \Big(\tau^{n}_1\oplus x_j\tau^{n}_2\Big)^{s_k-s_{j-1}}_{\mathcal{R}(p,q)}}=\Big(1 \oplus x_j\Big)^{n}_{\mathcal{R}(p,q)}
		\end{eqnarray}
	\end{small}
	and so 
	\begin{eqnarray*}
		s_n(x_1,\cdots,x_k;p,q)=\prod_{j=1}^{k}\Big(1 \oplus x_j\Big)^{n}_{\mathcal{R}(p,q)}.
	\end{eqnarray*}
The equivalently formula can be derived by putting $p=p^{-1}, q=q^{-1}, x_j=x^{-1}_j$ and ${\mathcal R}(p,q)=\mathcal {R}(p^{-1},q^{-1}).$ 
$\cqfd$
\end{proof}
The alternative $\mathcal{R}(p,q)-$ deformed multinomial formula is presented in the following theorem.
\begin{theorem}
	Let $x_j, j\in\{1,2,\cdots, k+1\}, p,$ and $q$ real numbers. For $n$ positive integer, the following result holds. 
	\begin{small}
		\begin{eqnarray}\label{th2.3}
		\big(1\ominus \Lambda_k\big)^{n}_{\mathcal{R}(p,q)}= \sum_{r_j=0}^{n}\bigg[\begin{array}{c}
		n\\ r_1,r_2,\cdots,r_k
		\end{array}\bigg]_{\mathcal{R}(p,q)}\prod_{j=1}^{k}x^{n-s_j}_j\big(1\ominus x_j\big)^{r_j}_{\mathcal{R}(p,q)}\big(1\ominus x_{k+1}\big)^{n-s_k}_{\mathcal{R}(p,q)},
		\end{eqnarray}
	\end{small}
	where $r_j\in\{0,\cdots,n\},$ $j\in\{1,\cdots,k\},$  with $\sum_{i=1}^{k}r_i\leq n$ and $s_j=\sum_{i=1}^{j}r_i,$ $s_0=0,$ and $\Lambda_k=\prod_{j=1}^{k+1}x_j.$
\end{theorem}
\begin{proof} From the ${\mathcal R}(p,q)-$ deformed binomial formula, we get
	\begin{eqnarray*}
		\Big(1 \ominus \Lambda_k\Big)^{n}_{\mathcal{R}(p,q)}&=& \sum_{r=0}^{n}\bigg[\begin{array}{c}
			n\\ r
		\end{array}\bigg]_{\mathcal{R}(p,q)}\tau^{n-r\choose 2}_1\tau^{r\choose 2}_2\Big(-\Lambda_k\Big)^r.
	\end{eqnarray*}
	Using the relation
	\begin{eqnarray*}
		\sum_{r_1=0}^{n-r}\,\bigg[\begin{array}{c}
			n-r\\ r_1
		\end{array}\bigg]_{\mathcal{R}(p,q)}\,x^{n-r-r_1}_1\,\Big(1\ominus x_1\Big)^{r_1}_{\mathcal{R}(p,q)}=1
	\end{eqnarray*}
	and interchanging the order of summation, we obtain:
	\begin{small}
		\begin{eqnarray*}
			\big(1\ominus \Lambda_k\big)^{n}_{\mathcal{R}(p,q)}
			&=&\sum_{r_1=0}^{n}\bigg[\begin{array}{c}
				n\\ r_1
			\end{array}\bigg]_{\mathcal {R}(p,q)}x^{n-r_1}_1\big(1\ominus x_1\big)^{r_1}_{\mathcal{R}(p,q)}\nonumber\\&\times& \sum_{r=0}^{n-r_1}\bigg[\begin{array}{c}
				n-r_1\\ r
			\end{array}\bigg]_{\mathcal{R}(p,q)}\tau^{n-r\choose 2}_1\tau^{r\choose 2}_2\big(-\Lambda_k\big)^r.
		\end{eqnarray*}
	\end{small}
	Once using the ${\mathcal R}(p,q)-$ deformed binomial formula, we have:
	\begin{eqnarray*}
		\big(1\ominus \Lambda_k\big)^{n}_{\mathcal{R}(p,q)}=\sum_{r_j=0}^{n}\bigg[\begin{array}{c}
			n\\ r_1
		\end{array}\bigg]_{\mathcal {R}(p,q)}\,x^{n-r_1}_1\Big(1\ominus x_1\Big)^{r_1}_{\mathcal{R}(p,q)}\big(1\ominus\Lambda_k\big)^{n-r_1}_{\mathcal{R}(p,q)}
	\end{eqnarray*}
	and generally,
	\begin{eqnarray*}
		\Big(1\ominus \Lambda_k\Big)^{n-s_{j-1}}_{\mathcal{R}(p,q)}=\sum_{r_j=0}^{n-s_{j-1}}\bigg[\begin{array}{c}
			n-s_{j-1}\\ r_j
		\end{array}\bigg]_{\mathcal {R}(p,q)}\,x^{n-s_j}_j\Big(1\ominus x_1\Big)^{r_1}_{\mathcal{R}(p,q)}\Big(1\ominus\Lambda_k\Big)^{n-s_j}_{\mathcal{R}(p,q)}
	\end{eqnarray*}
	for $j\in\{1,2,\cdots,k\}$ with $s_0=0.$ Applying the last expression, successively for $j\in\{1,2,\cdots,k\}$ and using the relation (\ref{id}), the result is immediately deduced. 
	$\cqfd$
\end{proof}
The results contained in the corollary below are the particular case of the relation (\ref{th2.3}) by taking $x_{k+1}=0.$
\begin{corollary}
	Let $n$ be a positive integer. Then, 
	\begin{eqnarray*}\label{cor2.1}
	\sum_{r_j=0}^{n}\bigg[\begin{array}{c}
	n\\ r_1,r_2,\cdots,r_k
	\end{array}\bigg]_{\mathcal {R}(p,q)}\prod_{j=1}^{k}\,x^{n-s_j}_j\big(1\ominus x_j\big)^{r_j}_{\mathcal{R}(p,q)}=\tau^{\frac{s_k(1+s_k-2n)}{2}}_1
	\end{eqnarray*}
	and
	\begin{eqnarray*}\label{cor2.2}
	\sum_{r_j=0}^{n}\bigg[\begin{array}{c}
	n\\ r_1,r_2,\cdots,r_k
	\end{array}\bigg]_{\mathcal{R}(p,q)}\prod_{j=1}^{k}\,x^{r_j}_j\big(1\ominus x_j\big)^{n-s_j}_{\mathcal{R}(p,q)}=\tau^{\frac{s_k(1+s_k-2n)}{ 2}}_1,
	\end{eqnarray*}
	where $j\in\{1,\cdots,k\},$  with $\displaystyle \sum_{i=1}^{k}r_j\leq n$ and $s_j=\sum_{i=1}^{j}r_i,$ $s_0=0.$
\end{corollary}

The generalization of the multinomial formula given by {\bf Gasper and Rahman} \cite{GR} can be determined as follows:
\begin{small}
	\begin{eqnarray*}
		\big(1\ominus \Lambda_k\big)^{n}_{\mathcal{R}(p,q)}&=& \sum_{r_j=0}^{n}\bigg[\begin{array}{c}
			n\\ r_1,r_2,\cdots,r_k
		\end{array}\bigg]_{\mathcal {R}(p,q)}\prod_{j=1}^{k}\,x^{s_j}_j\big(1\ominus x_{j-1}\big)^{n}_{\mathcal{R}(p,q)}\big(1\ominus x_{k}\big)^{n-s_k}_{\mathcal{R}(p,q)},
	\end{eqnarray*}
\end{small}
where $j\in\{1,2,\cdots, k\},$ with $\displaystyle \sum_{i=1}^{k}r_j\leq n$ and $s_j=\sum_{i=1}^{j}r_j.$
\section{$\mathcal{R}(p,q)-$ deformed multinomial distribution}
In this section, we construct the multinomial and negative multinomial probability distribution of the first and second kind in the framework of the $\mathcal{R}(p,q)-$ deformed quantum algebras. Moreover, the $\mathcal{R}(p,q)-$ deformed multiple Heine, Euler, negative multiple Heine, and negative Euler are obtained as limit of the above probability distriburion as $ n \longrightarrow\infty.$ We use the following notations in the sequel:
$\Theta=\big(\theta_1,\theta_2,\cdots,\theta_k\big).$
\subsection{$\mathcal{R}(p,q)-$ deformed multinomial distribution of the first kind}
We consider a sequence of independant Bernoulli trials with chain-composite successes (or failures) and suppose that the odds of success of the $j^{th}$ kind at the $i^{th}$ trial is furnished by:
\begin{eqnarray*}
	\theta_{j,i}=\theta_j\,\tau^{1-i}_1\,\tau^{i-1}_2,\quad 0< \theta_j<\infty,\quad (j,i)\in\mathbb{N}.
\end{eqnarray*}
The probability of success of the $j^{th}$ kind at the $i^{th}$ trial is derived as follows:
\begin{eqnarray}\label{ps}
p_{j,i}=\frac{\theta_j\,\tau^{i-1}_2}{  \tau^{i-1}_1 + \theta_j\,\tau^{i-1}_2}.
\end{eqnarray}
Naturally, the probability of failure of the $j^{th}$ kind at the $i^{th}$ trial is deduced as:
\begin{eqnarray}\label{pf}
q_{j,i}=\frac{\tau^{i-1}_1}{ \tau^{i-1}_1 + \theta_j\,\tau^{i-1}_2}.
\end{eqnarray}
Note that, taking $\mathcal{R}(x,1)=\frac{x-1}{1-q},$ we recovered the $q-$ deformation of probabilities (\ref{ps})and (\ref{pf}) given in \cite{CA2}:
\begin{eqnarray*}
p_{j,i}=\frac{\theta_j\,q^{i-1}}{1 + \theta_j\,q^{i-1}}\quad\mbox{and}\quad
q_{j,i}=\frac{1}{1 + \theta_j\,q^{i-1}}.
\end{eqnarray*}

We denote by $Y_j, j\in\{1,2,\cdots,k\}$ the number of successes of the $j^{th}$ kind in a sequence of $n$ independent Bernoulli trials with chain-composite failures and $\big(
Y_1,Y_2,\cdots,Y_k\big)$ the vector of the $\mathcal{R}(p,q)-$ deformed random variables.
\begin{theorem}\label{th1s3}
	The  $\mathcal{R}(p,q)-$ deformed multinomial probability distribution of the first kind with parameters $n, \Theta, p,$ and $q$ is presented by:
	\begin{small}
		\begin{eqnarray}\label{rpqmfk}
		P\big(Y_1=y_1,\cdots,Y_k=y_k\big)=\bigg[\begin{array}{c}
		n\\ y_1,y_2,\cdots,y_k
		\end{array}\bigg]_{\mathcal {R}(p,q)}\prod_{j=1}^{k}\frac{\theta^{y_j}_j\tau^{n-y_j\choose 2}_1\tau^{y_j\choose 2}_2}{(1 \oplus \theta_j)^{n-s_{j-1}}_{\mathcal{R}(p,q)}}
		\end{eqnarray}
		and	their recursion relations as:
		\begin{eqnarray*}
			P_{y+1}= \Big[n-\sum_{j=1}^{k}y_j\Big]_{k,\mathcal{R}(p,q)}\prod_{j=1}^{k}\frac{\theta_j\tau^{n-y_j}_1\tau^{y_j}_2P_{y}}{[y_{j}+1]_{\mathcal{R}(p,q)}\big(1\oplus \theta_j\big)_{\mathcal{R}(p,q)}},\,\mbox{with}\, P_0= \prod_{j=1}^{k}\frac{\tau^{n\choose 2}_1}{\big(1\oplus\theta_j\big)^{n}_{\mathcal{R}(p,q)}},
		\end{eqnarray*} 
	\end{small}
	where $y_j\in\{0,1,\cdots,n\}, \displaystyle\sum_{j=1}^{k}y_j\leq n, s_j=\displaystyle\sum_{i=1}^{j}y_j, 0<\theta_j<1$, and $j\in\{1,2,\cdots,k\}.$
\end{theorem}
\begin{proof} 
	The $\mathcal{R}(p,q)-$ deformed random variable $Y_1$ is defined on the sequence of $n$ independent Bernoulli trials with space $\omega=\{s_1,f_1\},$ follows the  $\mathcal {R}(p,q)-$ deformed binomial distribution of the first kind with probability function:
	\begin{eqnarray*}
		P\big(Y_1=y_1\big)=\bigg[\begin{array}{c}
			n\\ y_1
		\end{array}\bigg]_{\mathcal{R}(p,q)}\frac{\theta^{y_1}_1\tau^{n-y_1\choose 2}_1\tau^{y_1\choose 2}_2}{ \big(1 \oplus \theta_1\big)^{n}_{\mathcal{R}(p,q)}},\quad y_1\in\{0,1,\cdots,n\}.
	\end{eqnarray*}
	In the same way, the $\mathcal{R}(p,q)-$ deformed random variable $Y_k$ is defined on the sequence of $n-s_{k-1}$ independent Bernoulli trials, with conditional space $\omega=\{s_k,f_k\},$ obeys a ${\mathcal R}(p,q)-$ deformed binomial distribution of the first kind with probability distribution:
	\begin{eqnarray*}
		P\big(Y_k=y_k\mid Y_1=y_1,\cdots,Y_{k-1}=y_{k-1}\big)=\bigg[\begin{array}{c}
			n-s_{k-1}\\ y_k
		\end{array}\bigg]_{\mathcal {R}(p,q)}\frac{\theta^{y_k}_k\tau^{n-y_k\choose 2}_1\tau^{y_k\choose 2}_2}{\big(1 \oplus \theta_k\big)^{n-s_{k-1}}_{\mathcal{R}(p,q)}},
	\end{eqnarray*}
	where $ y_k\in\{0,1,\cdots,n-s_{k-1}\}.$ Then, from the relations (\ref{id}) and the multiplicative formula for probabilities, the result follows.
	Using the $\mathcal{R}(p,q)-$ deformed multinomial formula, we get:
	\begin{eqnarray*}
		\sum\,P\big(Y_1=y_1,\cdots,Y_k=y_k\big)=1.
	\end{eqnarray*}
The recurrence relation is obtained by simpler computation.
		$\cqfd$
\end{proof}

We consider $T_j, j\in\{1,2,\cdots,k\}$ the number of successes of the $j^{th}$ kind until the occurrence of the $n^{th}$ failure of the $k^{th}$ kind,  in a sequence of $n$ independent Bernoulli trials with chain-composite failures and $\big(T_1,T_2,\cdots, T_k\big)$ the $\mathcal{R}(p,q)-$ deformed random vector.
\begin{theorem}
	The probability function of the negative $\mathcal{R}(p,q)-$ deformed multinomial distribution of the first kind with parameters $n, \Theta, p$ and $q$ is given as follows:
	\begin{small}
		\begin{eqnarray}\label{nrpqmfk}
		P\big(T_1=t_1,\cdots,T_k=t_k\big)=\bigg[\begin{array}{c}
		n+s_k-1\\ t_1,t_2,\cdots,t_k
		\end{array}\bigg]_{\mathcal {R}(p,q)}\prod_{j=1}^{k}\frac{\theta^{t_j}_j\tau^{n-t_j\choose 2}_1\tau^{t_j\choose 2}_2}{\big(1 \oplus \theta_j\big)^{n+s_k-s_{j-1}}_{\mathcal{R}(p,q)}}
		\end{eqnarray}
		and their recurrence relation by:
		\begin{eqnarray*}
			P_{t+1}= \Big[n-\sum_{j=1}^{k}t_j\Big]_{k,\mathcal{R}(p,q)}\prod_{j=1}^{k}\frac{\theta_j\tau^{n-t_j}_1\tau^{t_j}_2P_{t}}{[t_{j}+1]_{\mathcal{R}(p,q)}\big(1\ominus \theta_j\big)_{\mathcal{R}(p,q)}},\,\mbox{with}\, P_0= \prod_{j=1}^{k}\frac{\tau^{n\choose 2}_1}{\big(1\oplus\theta_j\big)^{n}_{\mathcal{R}(p,q)}},
		\end{eqnarray*} 
	\end{small}
	where $t_j\in\mathbb{N}, s_j=\displaystyle\sum_{i=1}^{j}t_j, 0<\theta_j<1,$ and $j\in\{1,2,\cdots,k\}.$
\end{theorem}
\begin{proof} 
	From the multiplicative formula, we have:
	\begin{small} 
		\begin{eqnarray*}\label{nmfor}
		P\big(T_1=t_1,\cdots,T_k=t_k\big)&=&P\big(T_1=t_1\mid T_2=t_2,\cdots,T_{k}=t_{k} \big)\nonumber\\&\times&P\big(T_2=t_2\mid T_3=t_3,\cdots,T_{k}=t_{k}\big)\cdots P\big(T_k=t_k \big).
		\end{eqnarray*}
	\end{small}
	The $\mathcal{R}(p,q)-$ deformed random variable $T_1$ is defined on the sequence of $n+s_k$ independent Bernoulli trials with space $\Omega_1=\{s_1,f_1\},$ obeys the negative $\mathcal {R}(p,q)-$ deformed binomial distribution of the first kind with probability function
	\begin{eqnarray*}
		P\big(T_1=t_1\mid T_2=t_2,\cdots,T_{k}=t_{k}\big)=\genfrac{[}{]}{0pt}{}{n+s_k-1}{t_1}_{{\mathcal R}(p,q)}{\theta^{t_1}_1\tau^{{n+s_k-t_1 \choose 2}}_1\tau^{u_1\choose 2}_2\over \Big(1 \oplus \theta_1\Big)^{n+s_k}_{\mathcal{R}(p,q)}},\quad t_1\in\mathbb{N}.
	\end{eqnarray*}
	Then,  given the occurrence of the event $\{T_{k}=t_{k}\},$ the $\mathcal {R}(p,q)-$ deformed random variable $T_k$ is defined on the sequence of $n+s_k-s_{k-1}=n+t_k$ independent Bernoulli trials, with conditional space $\Omega_k=\{s_k,f_k\},$ obeys the negative $\mathcal {R}(p,q)-$ deformed binomial distribution of the first kind with probability distribution:
	\begin{small}
		\begin{eqnarray*}
			P\big(T_k=t_k\big)=\genfrac{[}{]}{0pt}{}{n+s_k-s_{k-1}-1}{u_k}_{\mathcal {R}(p,q)}{\theta^{u_k}_k\tau^{n+u_k\choose 2}_1\tau^{t_k\choose 2}_2\over \Big(1 \oplus \theta_k\Big)^{n+s_k-s_{k-1}}_{\mathcal{R}(p,q)}},
		\end{eqnarray*}
	\end{small}
	where $ t_k\in\mathbb{N}.$ Then, multplying all the above probabilities and using the relation
	\begin{eqnarray*}
		\genfrac{[}{]}{0pt}{}{n+s_k-1}{t_1,t_2,\ldots,t_k}_{{\mathcal R}(p,q)}=\prod_{j=1}^{k}\genfrac{[}{]}{0pt}{}{n+s_k-s_{j-1}-1}{t_j}_{\mathcal {R}(p,q)},\quad s_0=0
	\end{eqnarray*}
	the result follows. 
	Using the negative ${\mathcal R}(p,q)-$ deformed multinomial formula, we get
	\begin{eqnarray*}
		\sum\,P\big(T_1=t_1,\ldots,U_k=t_k\big)=1.
	\end{eqnarray*}.
		$\cqfd$
\end{proof} 
\begin{remark}
	We denote by $V_j, j\in\{1,2,\ldots,k\},$ the number of failures of the $j^{th}$ kind until the occurrence of the $n^{th}$ success of the $k^{th}$ kind, in a sequence of independent Bernoulli trials with chain-composite successes and $\big(V_1,V_2,\cdots,V_k\big)$ the $\mathcal{R}(p,q)-$ deformed random vector. The probability function of the negative $\mathcal {R}(p,q)-$ binomial distribution of the first kind is:
	\begin{eqnarray}\label{nrbfk}
	P\big(V=v\big)=\bigg[\begin{array}{c}
	n+v-1\\ v
	\end{array}\bigg]_{\mathcal {R}(p,q)}\frac{\theta^{n}\tau^{v \choose 2}_1\tau^{{n\choose 2}+v}_2}{ \Big(1 \oplus \theta_1\Big)^{n+v}_{\mathcal{R}(p,q)}},\quad v\in\mathbb{N}.
	\end{eqnarray}
	From the relation (\ref{nrbfk}) and the steps used to get the (\ref{nrpqmfk}), the probability function of the random vector $\big(V_1,V_2,\cdots,V_k\big)$ is given by:
	\begin{small}
		\begin{eqnarray}\label{anrmfk}
		P\big(V_1=v_1,\cdots,V_k=v_k\big)=\bigg[\begin{array}{c}
		n+s_k-1\\ v_1,v_2,\cdots,v_k
		\end{array}\bigg]_{\mathcal {R}(p,q)}\prod_{j=1}^{k}\frac{\theta^{n+s_k-s_{j-1}}_j\tau^{v_j\choose 2}_1\tau^{{n+s_k-s_{j-1}\choose 2}+v_j}_2}{\big(1 \oplus \theta_j\big)^{n+s_k-s_{j-1}}_{\mathcal{R}(p,q)}},
		\end{eqnarray}
	\end{small}
	where $v_j\in\mathbb{N}, s_j=\displaystyle\sum_{i=1}^{j}v_j, 0<\theta_j<1$, and $j\in\{1,2,\cdots,k\}.$
\end{remark}
\begin{remark}
	The $\mathcal{R}(p,q)-$ deformed multinomial distributions (\ref{rpqmfk}) and (\ref{nrpqmfk}) can be approximated by the probability function of the $\mathcal{R}(p,q)-$ deformed multiple Heine distributions (\ref{rpqH}) and (\ref{rpqH1}). In fact, setting $\mu_j=\frac{\theta_j}{\tau_1-\tau_2}$ and using $0<q<p<1,$ we have:
	\begin{eqnarray*}
		\lim_{n\longrightarrow \infty} \bigg[\begin{array}{c}
			n\\ y_1,y_2,\cdots,y_k
		\end{array}\bigg]_{\mathcal {R}(p,q)}=\frac{1}{\prod^{k}_{j=1}\big(\tau_1-\tau_2\big)^{y_j}[y_j]_{\mathcal{R}(p,q)}!}
	\end{eqnarray*}
	and
	\begin{eqnarray*}
		\lim_{n\longrightarrow\infty}\Big(1 \oplus \mu_j(\tau_1-\tau_2)\Big)^{n-s_{j-1}}_{\mathcal{R}(p,q)}=\frac{1}{e_{\mathcal{R}(p,q)}(-\mu_j)}.
	\end{eqnarray*}
	Thus, 
	\begin{eqnarray}\label{rpqH}
	\lim_{n\longrightarrow \infty }\bigg[\begin{array}{c}
	n\\ y_1,y_2,\cdots,y_k
	\end{array}\bigg]_{\mathcal {R}(p,q)}\prod_{j=1}^{k}\frac{\theta^{x_j}_j\tau^{n-y_j\choose 2}_1\tau^{y_j\choose 2}_2}{(1 \oplus \theta_j)^{n-s_{j-1}}_{\mathcal{R}(p,q)}}=\prod_{j=1}^{k}e_{\mathcal{R}(p,q)}(-\mu_j)\frac{\mu^{y_j}_j\tau^{y_j\choose 2}_2}{[y_j]_{\mathcal{R}(p,q)}!}.
	\end{eqnarray}
	Similarly, we get:
	\begin{small}
		\begin{eqnarray}\label{rpqH1}
		\lim_{n\longrightarrow \infty }\bigg[\begin{array}{c}
		n+s_k-1\\ tt_1,t_2,\cdots,t_k
		\end{array}\bigg]_{\mathcal {R}(p,q)}\prod_{j=1}^{k}\frac{\theta^{t_j}_j\tau^{n-t_j\choose 2}_1\tau^{t_j\choose 2}_2}{\big(1 \oplus \theta_j\big)^{n+s_k-s_{j-1}}_{\mathcal{R}(p,q)}}=\prod_{j=1}^{k}E_{\mathcal{R}(p,q)}(-\mu_j)\frac{\mu^{t_j}_j\tau^{t_j\choose 2}_2}{[t_j]_{\mathcal{R}(p,q)}!}.
		\end{eqnarray}
	\end{small}
\end{remark}
\subsection{$\mathcal{R}(p,q)-$ deformed multinomial distribution of the second kind}
We consider a sequence of independent Bernoulli trials with chain-composite successes(or failures) and suppose that the conditional probability of success of the $j^{th}$ kind at any trial, given that $i-1$ successes of the $j^{th}$ kind occur in the previous trials, is given by:
\begin{eqnarray}\label{4.1}
p_{j,i}=1-\theta_j\,\tau^{1-i}_1\,\tau^{i-1}_2,\quad 0<\theta_j<1,\quad (j,i)\in\mathbb{N}.
\end{eqnarray}
We denote by $X_{j}$ the number of failures of the $j^{th}$ kind in a sequence of $n$ independent Bernoulli trials with chain-composite successes, where the conditional probability of success of the $j^{th}$ kind at any trial, given that $i-1$ successes of the $j^{th}$ kind occur in the previous trials, is given by (\ref{4.1}).
\begin{theorem}
	The probability function of the $\mathcal {R}(p,q)-$ deformed multinomial distribution of the second kind with parameters $n,\Theta, p$ and $q$ is determined by:
	\begin{small}
		\begin{eqnarray}\label{rpqmsk}
		P_x:=P\big(X_1=x_1,\cdots,X_k=x_k\big)=\bigg[\begin{array}{c}
		n\\ x_1,x_2,\cdots,x_k
		\end{array}\bigg]_{\mathcal {R}(p,q)}\prod_{j=1}^{k}\theta^{x_j}_j\,\Big(1\ominus\theta_j\Big)^{n-s_{j}}_{\mathcal{R}(p,q)}.
		\end{eqnarray}
	\end{small}
	The recurrence relation for the $\mathcal{R}(p,q)-$ deformed multinomial distribution of the second kind is given by:
	\begin{small}
		\begin{eqnarray*}
			P_{x+1}= \Big[n-\sum_{j=1}^{k}x_j\Big]_{k,\mathcal{R}(p,q)}\prod_{j=1}^{k}\frac{\theta_j\big(1\ominus \theta_j\big)_{\mathcal{R}(p,q)}}{[x_{j}+1]_{\mathcal{R}(p,q)}}P_{x},\quad\mbox{with}\quad P_0= \prod_{j=1}^{k}\big(1\ominus\theta_j\big)^{n}_{\mathcal{R}(p,q)}.
		\end{eqnarray*} 
		where $x_j\in\{0,1,\cdots,n\}, \displaystyle\sum_{j=1}^{k}x_j\leq n, s_j=\displaystyle\sum_{i=1}^{j}x_j, 0<\theta_j<1$, and $j\in\{1,2,\cdots,k\}.$
	\end{small}
\end{theorem}
\begin{remark}
	\begin{enumerate}
		\item[(i)] Taking $k=1,$ we deduced the $\mathcal{R}(p,q)-$ deformed binomial distribtuion of the second kind \cite{HMD}.
		\item[(i)] The multinomial probability distribution presented in \cite{O} is recovered by putting $\mathcal{R}(p,q)=1.$
	\end{enumerate}
\end{remark}
\begin{corollary}
	The recursion relation for the $q-$ deformed multinomial distribution of the second kind is deduced as :
	\begin{eqnarray*}
		P_{x+1}= \Big[n-\sum_{j=1}^{k}x_j\Big]_{k,q}\prod_{j=1}^{k}\frac{\theta_j\,(1-\theta_j)(1-q)}{1-q^{x_j+1}}\,P_{x}.
	\end{eqnarray*}
\end{corollary}
\begin{proof}
	By taking $\mathcal{R}(x,1)=\frac{1-x}{1-q}$ in the general formalism. 
	$\cqfd$
\end{proof}
\begin{remark}
	We denote by $Y_{j}, j\in\{1,2,\cdots,k\},$ the number of usccesses of the $j^{th}$ kind in a sequence of $n$ independent Bernoulli trials with chain-composite failures, where the conditional probability of success of the $j^{th}$ kind at any trial, given that $i-1$ successes of the $j^{th}$ kind occur in the previous trials, is given by the relation (\ref{4.1}). 
	
	Using the same procedure to derive the relation (\ref{rpqmsk}), the probability function of the random vector $\big(Y_1,Y_2,\cdots,Y_k\big)$ is obtained as:
	\begin{small}
		\begin{eqnarray}\label{rpqmsk1}
		P\big(Y_1=y_1,\cdots,Y_k=y_k\big)=\bigg[\begin{array}{c}
		n\\ y_1,y_2,\cdots,y_k
		\end{array}\bigg]_{\mathcal {R}(p,q)}\prod_{j=1}^{k}\theta^{n-s_j}_j\big(1\ominus\theta_j\big)^{y_j}_{\mathcal{R}(p,q)},
		\end{eqnarray}
	\end{small}
	where $y_j\in\{0,1,\cdots,n\}, \displaystyle\sum_{j=1}^{k}y_j\leq n, s_j=\displaystyle\sum_{i=1}^{j}y_j, 0<\theta_j<1$, and $j\in\{1,2,\cdots,k\}.$
\end{remark}

Let $W_j, j\in\{1,2,\ldots,k\}$ be the number of failures of the $j^{th}$ kind until the occurrence of the $n^{th}$ success of the $k^{th}$ kind,  in a sequence of $n$ independent Bernoulli trials with chain-composite successes.
\begin{theorem}
	The probability function of the negative ${\mathcal R}(p,q)-$ deformed multinomial distribution of the second kind with parameters $n,\Theta, p$ and $q$ is furnished by:
	\begin{small}
		\begin{eqnarray}\label{nrpqmsk}
		P_w:=P\big(W_1=w_1,\cdots,W_k=w_k\big)=\bigg[\begin{array}{c}
		n+s_k-1\\ w_1,\cdots,w_k
		\end{array}\bigg]_{\mathcal{R}(p,q)}\prod_{j=1}^{k}{\theta^{w_j}_j \big(1\ominus\theta_j\big)^{n+s_k-s_{j}}_{\mathcal{R}(p,q)}},
		\end{eqnarray}
		where $w_j\in\mathbb{N}, s_j=\displaystyle\sum_{i=1}^{j}w_j, 0<\theta_j<1$, and $j\in\{1,2,\cdots,k\}.$
		Furthermore, their recursion relations are given by:
		\begin{eqnarray*}
			P_{x+1}= \Big[n-\sum_{j=1}^{k}x_j\Big]_{k,\mathcal{R}(p,q)}\prod_{j=1}^{k}\frac{\theta_j\,\big(1\ominus \theta_j\big)_{\mathcal{R}(p,q)}}{[x_{j}+1]_{\mathcal{R}(p,q)}}P_{x},\,\mbox{with}\quad P_0= \prod_{j=1}^{k}\big(1\ominus\theta_j\big)^{n}_{\mathcal{R}(p,q)}.
		\end{eqnarray*}
	\end{small}
\end{theorem}
\begin{remark}
	The limit of the $\mathcal {R}(p,q)-$ deformed multinomial distribution of the second kind  (\ref{rpqmsk}), as $n\longrightarrow \infty$ is  the $\mathcal{R}(p,q)-$ deformed multiple Euler distribution:
	\begin{eqnarray*}
		\lim_{n\longrightarrow\infty} \bigg[\begin{array}{c}
			n\\ x_1,x_2,\cdots,x_k
		\end{array}\bigg]_{\mathcal {R}(p,q)}\prod_{j=1}^{k}\theta^{x_j}_j\,\Big(1\ominus\theta_j\Big)^{n-s_{j}}_{\mathcal{R}(p,q)}=\prod_{j=1}^{k}E_{\mathcal{R}(p,q)}(-\mu_j)\frac{\mu^{x_j}_j}{[x_j]_{\mathcal{R}(p,q)}}.
	\end{eqnarray*}
	Moreover, 
	The limit of the $\mathcal {R}(p,q)-$ deformed multinomial distribution of the second kind  (\ref{nrpqmsk}), as $n\longrightarrow \infty$ is  the $\mathcal{R}(p,q)-$ deformed multiple Euler distribution:
	\begin{eqnarray*}
		\lim_{n\longrightarrow\infty} \bigg[\begin{array}{c}
			n+s_k-1\\ w_1,w_2,\cdots,w_k
		\end{array}\bigg]_{\mathcal{R}(p,q)}\prod_{j=1}^{k}{\theta^{w_j}_j \big(1\ominus\theta_j\big)^{n+s_k-s_{j}}_{\mathcal{R}(p,q)}}=\prod_{j=1}^{k}E_{\mathcal{R}(p,q)}(-\mu_j)\frac{\mu^{w_j}_j}{[w_j]_{\mathcal{R}(p,q)}}.
	\end{eqnarray*}
\end{remark}
\begin{remark}
	Several kind of the ${\mathcal R}(p,q)-$ deformed multivariate absorption distribution are also attracted our attention.
	Replacing $\mathcal {R}(p,q)$ by $\mathcal {R}(p^{-1},q^{-1}),$ $\theta_j$ by $\tau^{-m_j}_1\tau^{m_j}_2,$ for $j\in\{1,2,\cdots,k\}$ in the relation (\ref{4.1}), the probability of successes is reduced as:  
	\begin{eqnarray*}
		p_{j,i}=1-\tau^{-m_ji+1}_1\tau^{m_j+1-i}_2,\quad 0<m_j<\infty,\quad j\in\{1,2,\ldots,k\},\quad i\in\{1,2,\cdots,[m_j]\}.
	\end{eqnarray*}
	Using the relation (\ref{eq2.2}) and the $\mathcal{R}(p,q)-$ deformed factorial, the probability function (\ref{rpqmsk}) takes the following form:
	\begin{small}
		\begin{eqnarray*}\label{rmad}
		P\big(X_1=x_1,\cdots,X_k=x_k\big)
		&=&\bigg[\begin{array}{c}
		n\\ x_1,x_2,\cdots,x_k
		\end{array}\bigg]_{\mathcal{R}(p,q)}(\tau_1\,\tau_2)^{-\displaystyle\sum_{j=1}^k x_j(m_j-n +s_j)}\nonumber\\&\times&\prod_{j=1}^{k}\big(\tau_1-\tau_2\big)^{n-s_j}[m_j]_{n-s_j,\mathcal{R}(p,q)}.
		\end{eqnarray*}
		Furtermore, from (\ref{eq2.2}), the probability function (\ref{rpqmsk1}) can be rewritten as:
		\begin{eqnarray*}
		P\big(Y_1=y_1,\cdots,Y_k=y_k\big)&=&\bigg[\begin{array}{c}
		n\\ y_1,y_2,\cdots,y_k
		\end{array}\bigg]_{\mathcal {R}(p,q)}(\tau_1\,\tau_2)^{-\displaystyle\sum_{j=1}^k (m_j-y_j)(n -s_j)}\nonumber\\&\times&\prod_{j=1}^{k}\big(\tau_1-\tau_2\big)^{y_j}[m_j]_{y_j,\mathcal{R}(p,q)}.
		\end{eqnarray*}
	\end{small}
\end{remark}
\section{Particular cases of multinomial distribution}
In this section, we derive particular multinomial coefficient and multinomial probability distribution induced by the quantum algebras known in the literature.
	\begin{enumerate}
	\item[(i)]	Taking $\mathcal{R}(x)=\frac{x-x^{-1}}{q-q^{-1}},$ we obtain the results from the {\bf Biederhan-Macfarlane} algebra\cite{BC,M}: the {$q$- deformed multinomial coefficient}
	\begin{eqnarray*}
		\bigg[\begin{array}{c}x\\r_1,r_2,\cdots,r_k\end{array}\bigg]_{q}=\frac{[x]_{r_1+r_2+\cdots+r_k,q}}{ [r_1]_{q}![r_2]_{q}!\cdots[r_k]_{q}!}
	\end{eqnarray*}
	satisfies the recursion relation:
	\begin{small}
		\begin{eqnarray*}
			&&\bigg[\begin{array}{c}x\\r_1,r_2,\cdots,r_k\end{array}\bigg]_{q}=q^{s_k}\,\bigg[\begin{array}{c}x-1\\r_1,r_2,\cdots,r_k\end{array}\bigg]_{q}
			+q^{-x+m_1}\bigg[\begin{array}{c}x-1\\r_1-1,r_2,\cdots,r_k\end{array}\bigg]_{q}\cr
			&&\qquad\qquad+q^{-x+m_2}\bigg[\begin{array}{c}x-1\\r_1,r_2-1,\cdots,r_k\end{array}\bigg]_{q}
			+\cdots+q^{-x+m_k}\bigg[\begin{array}{c}x-1\\r_1,r_2,\cdots,r_{k-1}\end{array}\bigg]_{q}
		\end{eqnarray*}
		and alternatively,
		\begin{eqnarray*}
			&&\bigg[\begin{array}{c}x\\r_1,r_2,\cdots,r_k\end{array}\bigg]_{q}=q^{-s_k}\bigg[\begin{array}{c}x-1\\r_1,r_2,\cdots,r_k\end{array}\bigg]_{q}
			+q^{x-m_1}\bigg[\begin{array}{c}x-1\\r_1-1,r_2,\cdots,r_k\end{array}\bigg]_{q}\cr
			&&\qquad\qquad+q^{x-m_2}q^{-s_1}\bigg[\begin{array}{c}x-1\\r_1,r_2-1,\cdots,r_k\end{array}\bigg]_{q}
			+\cdots+q^{x-m_k}q^{-s_{k-1}}\bigg[\begin{array}{c}x-1\\r_1,r_2,\cdots,r_{k-1}\end{array}\bigg]_{q}.
		\end{eqnarray*}
	\end{small}
	Moreover, 
	the $q^{-1}$- deformed multinomial coefficient provided by
	\begin{eqnarray*}
		\bigg[\begin{array}{c}x\\r_1,r_2,\cdots,r_k\end{array}\bigg]_{q^{-1}}
		&=&\bigg[\begin{array}{c}x\\r_1,r_2,\cdots,r_k\end{array}\bigg]_{q}\nonumber\\
		&=&\bigg[\begin{array}{c}x\\r_1,r_2,\cdots,r_k\end{array}\bigg]_{q}
	\end{eqnarray*}
	obey the recursion relation: 
	\begin{small}
		\begin{eqnarray*}
			&&\bigg[\begin{array}{c}x\\r_1,r_2,\cdots,r_k\end{array}\bigg]_{q}=q^{-m_1}\bigg[\begin{array}{c}x-1\\r_1,r_2,\cdots,r_k\end{array}\bigg]_{q}
			+q^{-m_2}\bigg[\begin{array}{c}x-1\\r_1-1,r_2,\cdots,r_k\end{array}\bigg]_{q}\cr&
			&\qquad\qquad\qquad+q^{-m_3}\bigg[\begin{array}{c}x-1\\r_1,r_2-1,\cdots,r_k\end{array}\bigg]_{q}
			+\cdots+ q^x\bigg[\begin{array}{c}x-1\\r_1,r_2,\cdots,r_{k-1}\end{array}\bigg]_{q}.
		\end{eqnarray*}
	\end{small}
	and
	\begin{small}
		\begin{eqnarray*} 
			&&\bigg[\begin{array}{c}x\\r_1,r_2,\cdots,r_k\end{array}\bigg]_{p,q}=p^x\bigg[\begin{array}{c}x-1\\r_1,r_2,\cdots,r_k\end{array}\bigg]_{p,q}
			+q^{x-s_1}\bigg[\begin{array}{c}x-1\\r_1-1,r_2,\cdots,r_k\end{array}\bigg]_{p,q}\cr
			&&\qquad\qquad\qquad+q^{x-s_2}\bigg[\begin{array}{c}x-1\\r_1,r_2-1,\cdots,r_k\end{array}\bigg]_{p,q}
			+\cdots+q^{x-s_k}\bigg[\begin{array}{c}x-1\\r_1,r_2,\cdots,r_{k-1}\end{array}\bigg]_{p,q},
		\end{eqnarray*}
	\end{small}
	where $r_j\in\mathbb{N}$ and $j\in\{1,2,\cdots,k\},$ with $m_j=\sum_{i=j}^kr_i$ and $s_j=\sum_{i=1}^jr_i.$
	
	For $n$ a positive integers, $x,$ and $q$ real numbers, the following relation holds:
	\begin{small}
		\begin{eqnarray*}
			\prod_{j=1}^{k}\big(1\oplus x_j\big)^{n}_{q}=\sum\bigg[\begin{array}{c}n\\r_1,r_2,\cdots,r_k\end{array}\bigg]_{q}\prod_{j=1}^{k}x^{r_j}_jq^{n-r_j\choose 2}q^{-{r_j\choose 2}}\Big(q^{n-s_{j-1}}\oplus x_jq^{-n+s_{j-1}}\Big)^{s_{j-1}}_{q},
		\end{eqnarray*}
	\end{small}
	where  $r_j\in\{0,\cdots,n\},$ $j\in\{1,\cdots,k\},$  with $\sum_{i=1}^{k}r_i\leq n$ and $s_j=\sum_{i=1}^{j}r_i,$ $s_0=0.$
	
	Furthermore, for  $n$ be a positive integers, we have:
	\begin{small}
		\begin{eqnarray*}
			\prod_{j=1}^{k}\big(1\oplus x_j\big)^{n}_{q}=\sum\bigg[\begin{array}{c}n+s_{k}-1\\r_1,r_2,\cdots,r_k\end{array}\bigg]_{q}\prod_{j=1}^{k}\frac{x^{r_j}_jq^{n-r_j\choose 2}q^{-{r_j\choose 2}}}{ \big(q^{n}\oplus x_jq^{-n}\big)^{s_k-s_{j-1}}_{q}}.
		\end{eqnarray*}
		Equivalently,
		\begin{eqnarray*}
			\prod_{j=1}^{k}\big(1\oplus x_j\big)^{n}_{q}&=&\sum_{r_j\in\mathbb{N}}\bigg[\begin{array}{c}n+s_{k}-1\\r_1,r_2,\cdots,r_k\end{array}\bigg]_{q}\prod_{j=1}^{k}\frac{x^{n+s_k-s_{j-1}}_jq^{n-r_j\choose 2}q^{-{n+s_k-s_{j-1}\choose 2}-r_j}}{ \big(p^{n}\oplus x_jq^{-n}\big)^{s_k-s_{j-1}}_{q}},
		\end{eqnarray*} 
		where $j\in\{1,2,\cdots,k\},$ with $s_j=\displaystyle\sum_{i=1}^{j}r_j,\quad s_0=0.$
	\end{small}
	
	Let $x_j, j\in\{1,2,\cdots, k+1\}, $ and $q$ real numbers. For $n$ positive integer, the following result holds. 
	\begin{small}
		\begin{eqnarray*}
			\big(1 \ominus \Lambda_k\big)^{n}_{q}= \sum_{r_j=0}^{n}\bigg[\begin{array}{c}n\\r_1,r_2,\cdots,r_k\end{array}\bigg]_{q}\prod_{j=1}^{k}x^{n-s_j}_j\big(1\ominus x_j\big)^{r_j}_{q}\big(1\ominus x_{k+1}\big)^{n-s_k}_{q},
		\end{eqnarray*}
	\end{small}
	where $r_j\in\{0,\cdots,n\},$ $j\in\{1,\cdots,k\},$  with $\sum_{i=1}^{k}r_i\leq n$ and $s_j=\sum_{i=1}^{j}r_i,$ $s_0=0,$ $\Lambda_k=\prod_{j=1}^{k+1}x_j.$
	
	For $n$  a positive integer, we have: 
	\begin{eqnarray*}
		\sum_{r_j=0}^{n}\bigg[\begin{array}{c}n\\r_1,r_2,\cdots,r_k\end{array}\bigg]_{q}\prod_{j=1}^{k}\,x^{n-s_j}_j\big(1\ominus x_j\big)^{r_j}_{q}=q^{\frac{s_k(1+s_k-2n)}{2}}
	\end{eqnarray*}
	and
	\begin{eqnarray*}
		\sum_{r_j=0}^{n}\bigg[\begin{array}{c}n\\r_1,r_2,\cdots,r_k\end{array}\bigg]_{q}\prod_{j=1}^{k}\,x^{r_j}_j\big(1\ominus x_j\big)^{n-s_j}_{q}=q^{\frac{s_k(1+s_k-2n)}{2}},
	\end{eqnarray*}
	where $j\in\{1,\cdots,k\},$  with $\displaystyle \sum_{i=1}^{k}r_j\leq n$ and $s_j=\sum_{i=1}^{j}r_i,$ $s_0=0.$
	
	The $q-$ deformed of the multinomial formula given by {\bf Gasper and Rahman} \cite{GR} can be determined as follows:
	\begin{small}
		\begin{eqnarray*}
			\big(1\ominus \Lambda_k\big)^{n}_{q}&=& \sum_{r_j=0}^{n}\bigg[\begin{array}{c}n\\r_1,r_2,\cdots,r_k\end{array}\bigg]_{q}\prod_{j=1}^{k}\,x^{s_j}_j\big(1\ominus x_{j-1}\big)^{n}_{q}\big(1\ominus x_{k}\big)^{n-s_k}_{q},
		\end{eqnarray*}
	\end{small}
	where $j\in\{1,2,\cdots, k\},$ with $\displaystyle \sum_{i=1}^{k}r_j\leq n$ and $s_j=\sum_{i=1}^{j}r_j.$
	\begin{itemize}
		\item[(a)] 
		The probability function of the $q-$ deformed multinomial distribution of the first kind with parameters $n, \big(\theta_1,\theta_2,\cdots,\theta_k\big),$ and $q$  
		is presented by:
		\begin{eqnarray*}
			P\big(Y_1=y_1,\cdots,Y_k=y_k\big)=\genfrac{[}{]}{0pt}{}{n}{y_1,y_2,\cdots,y_k}_{q}\prod_{j=1}^{k}\frac{\theta^{y_j}_j\,q^{n-x_j\choose 2}\,q^{-{y_j\choose 2}}}{\big(1\oplus \theta_j\big)^{n-s_{j-1}}_{q}},
		\end{eqnarray*}
		and	their recursion relations as:
		\begin{eqnarray*}
			P_{y+1}= \Big[n-\sum_{j=1}^{k}y_j\Big]_{k,q}\prod_{j=1}^{k}\frac{\theta_j\,q^{n-y_j}q^{-y_j}P_{y}}{[y_{j}+1]_{q}\big(1\oplus \theta_j\big)_{q}},\,\mbox{with}\, P_0= \prod_{j=1}^{k}\frac{q^{n\choose 2}}{\big(1\oplus\theta_j\big)^{n}_{q}},
		\end{eqnarray*} 
		where $y_j\in\{0,1,\cdots,n\}, \displaystyle\sum_{j=1}^{k}y_j\leq n, s_j=\displaystyle\sum_{i=1}^{j}y_j, 0<\theta_j<1$, and $j\in\{1,2,\cdots,k\}.$
		\item[(b)]
		The probability function of the negative $q-$ deformed multinomial distribution of the first kind with parameters $n, \big(\theta_1,\theta_2,\cdots,\theta_k\big), $ and $q$ is given as follows:
		\begin{small}
			\begin{eqnarray*}
				P\big(T_1=t_1,\cdots,T_k=t_k\big)=\genfrac{[}{]}{0pt}{}{n+s_k-1}{t_1,t_2,\cdots,t_k}_{q}\prod_{j=1}^{k}\frac{\theta^{u_j}_j\,q^{n-u_j\choose 2}\,q^{-{u_j\choose 2}}}{\big(1\oplus \theta_j\big)^{n+s_k-s_{j-1}}_{q}},
			\end{eqnarray*}
			and their recurrence relation by:
			\begin{eqnarray*}
				P_{t+1}= \Big[n-\sum_{j=1}^{k}t_j\Big]_{k,q}\prod_{j=1}^{k}\frac{\theta_j\,q^{n-t_j}\,q^{-t_j}P_{t}}{[t_{j}+1]_{q}\big(1\ominus \theta_j\big)_{q}},\,\mbox{with}\, P_0= \prod_{j=1}^{k}\frac{q^{n\choose 2}}{\big(1\oplus\theta_j\big)^{n}_{q}},
			\end{eqnarray*}
		\end{small}
		where $t_j\in\mathbb{N}, s_j=\displaystyle\sum_{i=1}^{j}t_j, 0<\theta_j<1$, and $j\in\{1,2,\cdots,k\}.$
		\item[(c)] The probability function of the $q-$ deformed multinomial distribution of the second kind with parameters $n, \big(\theta_1,\theta_2,\cdots,\theta_k\big), $ and $q$ is determined by:
		\begin{eqnarray*}
			P\big(X_1=x_1,\cdots,X_k=x_k\big)=\genfrac{[}{]}{0pt}{}{n}{x_1,x_2,\cdots,x_k}_{q}\prod_{j=1}^{k}\theta^{x_j}_j\big(1\ominus \theta_j\big)^{n-s_{j}}_{q}
		\end{eqnarray*}
		and the recurrence relation 
		\begin{eqnarray*}
			P_{x+1}= \Big[n-\sum_{j=1}^{k}x_j\Big]_{k,q}\prod_{j=1}^{k}\frac{\theta_j\big(1\ominus \theta_j\big)_{q}}{[x_{j}+1]_{q}}P_{x},\quad\mbox{with}\quad P_0= \prod_{j=1}^{k}\big(1\ominus\theta_j\big)^{n}_{q}.
		\end{eqnarray*}	where $x_j\in\{0,1,\cdots,n\}, \displaystyle\sum_{j=1}^{k}x_j\leq n, s_j=\displaystyle\sum_{i=1}^{j}x_j.$
		
		{Another $q-$ deformed multinomial distribution of the second kind}
		\begin{eqnarray*}
			P\big(Y_1=y_1,\cdots,Y_k=y_k\big)=\bigg[\begin{array}{c} n  \\ y_1,y_2,\cdots,y_k\end{array} \bigg]_{q}\prod_{j=1}^{k}\theta^{n-s_j}_j\big(1\ominus \theta_j\big)^{y_j}_{q},
		\end{eqnarray*}
		where $y_j\in\{0,1,\cdots,n\}, \displaystyle\sum_{j=1}^{k}y_j\leq n, s_j=\displaystyle\sum_{i=1}^{j}y_j, 0<\theta_j<1$, and $j\in\{1,2,\cdots,k\}.$
		\item[(d)] The probability function of the negative $q-$ deformed multinomial distribution of the second kind with parameters $n, \big(\theta_1,\theta_2,\cdots,\theta_k\big), $ and $q$ is furnished by:
		\begin{small}
			\begin{eqnarray*}
				P\big(W_1=w_1,\cdots,W_k=w_k\big)=\bigg[\begin{array}{c} n+s_k-1  \\ w_1,w_2,\cdots,w_k\end{array} \bigg]_{q}\prod_{j=1}^{k}{\theta^{w_j}_j \big(1\ominus \theta_j\big)^{n+s_k-s_{j}}_{q}}.
			\end{eqnarray*}
			Moreover, their recursion relations are given as follows:
			\begin{eqnarray*}
				P_{x+1}= \Big[n-\sum_{j=1}^{k}x_j\Big]_{k,q}\prod_{j=1}^{k}\frac{\theta_j\,\big(1\ominus \theta_j\big)_{q}}{[x_{j}+1]_{q}}P_{x},\,\mbox{with}\quad P_0= \prod_{j=1}^{k}\big(1\ominus\theta_j\big)^{n}_{q}.
			\end{eqnarray*}
		\end{small}
		where $w_j\in\mathbb{N}, s_j=\displaystyle\sum_{i=1}^{j}w_j, 0<\theta_j<1$, and $j\in\{1,2,\cdots,k\}.$
	\end{itemize}
	\item[(ii)] Taking $\mathcal{R}(x,y)=\frac{x-y}{p-q},$ we obtain the results from the {\bf Jagannathan-Srinivassa} algebra\cite{JS}: the {$(p,q)$- deformed multinomial coefficient}
\begin{eqnarray*}
	\bigg[\begin{array}{c}x\\r_1,r_2,\cdots,r_k\end{array}\bigg]_{p,q}=\frac{[x]_{r_1+r_2+\cdots+r_k,p,q}}{ [r_1]_{p,q}![r_2]_{p,q}!\cdots[r_k]_{p,q}!}
\end{eqnarray*}
satisfies the recursion relation:
\begin{small}
	\begin{eqnarray*}
		&&\bigg[\begin{array}{c}x\\r_1,r_2,\cdots,r_k\end{array}\bigg]_{p,q}=p^{s_k}\,\bigg[\begin{array}{c}x-1\\r_1,r_2,\cdots,r_k\end{array}\bigg]_{p,q}
		+q^{x-m_1}\bigg[\begin{array}{c}x-1\\r_1-1,r_2,\cdots,r_k\end{array}\bigg]_{p,q}\cr
		&&\qquad\qquad+q^{x-m_2}\bigg[\begin{array}{c}x-1\\r_1,r_2-1,\cdots,r_k\end{array}\bigg]_{p,q}
		+\cdots+q^{x-m_k}\bigg[\begin{array}{c}x-1\\r_1,r_2,\cdots,r_{k-1}\end{array}\bigg]_{p,q}
	\end{eqnarray*}
	and alternatively,
	\begin{eqnarray*}
		&&\bigg[\begin{array}{c}x\\r_1,r_2,\cdots,r_k\end{array}\bigg]_{p,q}=q^{s_k}\bigg[\begin{array}{c}x-1\\r_1,r_2,\cdots,r_k\end{array}\bigg]_{p,q}
		+p^{x-m_1}\bigg[\begin{array}{c}x-1\\r_1-1,r_2,\cdots,r_k\end{array}\bigg]_{p,q}\cr
		&&\qquad\qquad+p^{x-m_2}q^{s_1}\bigg[\begin{array}{c}x-1\\r_1,r_2-1,\cdots,r_k\end{array}\bigg]_{p,q}
		+\cdots+p^{x-m_k}q^{s_{k-1}}\bigg[\begin{array}{c}x-1\\r_1,r_2,\cdots,r_{k-1}\end{array}\bigg]_{p,q}.
	\end{eqnarray*}
\end{small}
Moreover, 
the $(p^{-1},q^{-1})$- deformed multinomial coefficient provided by
\begin{eqnarray*}
	\bigg[\begin{array}{c}x\\r_1,r_2,\cdots,r_k\end{array}\bigg]_{p^{-1},q^{-1}}
	&=&(pq)^{-\displaystyle\sum_{j=1}^k r_j(x-m_j)}\bigg[\begin{array}{c}x\\r_1,r_2,\cdots,r_k\end{array}\bigg]_{p,q}\nonumber\\
	&=&(pq)^{-\displaystyle\sum_{j=1}^k r_j(x-s_j)}\bigg[\begin{array}{c}x\\r_1,r_2,\cdots,r_k\end{array}\bigg]_{p,q}
\end{eqnarray*}
obey the recursion relation: 
\begin{small}
	\begin{eqnarray*}
		&&\bigg[\begin{array}{c}x\\r_1,r_2,\cdots,r_k\end{array}\bigg]_{p,q}=q^{m_1}\bigg[\begin{array}{c}x-1\\r_1,r_2,\cdots,r_k\end{array}\bigg]_{p,q}
		+q^{m_2}\bigg[\begin{array}{c}x-1\\r_1-1,r_2,\cdots,r_k\end{array}\bigg]_{p,q}\cr&
		&\qquad\qquad\qquad+q^{m_3}\bigg[\begin{array}{c}x-1\\r_1,r_2-1,\cdots,r_k\end{array}\bigg]_{p,q}
		+\cdots+ p^x\bigg[\begin{array}{c}x-1\\r_1,r_2,\cdots,r_{k-1}\end{array}\bigg]_{p,q}.
	\end{eqnarray*}
\end{small}
and
\begin{small}
	\begin{eqnarray*} 
		&&\bigg[\begin{array}{c}x\\r_1,r_2,\cdots,r_k\end{array}\bigg]_{p,q}=p^x\bigg[\begin{array}{c}x-1\\r_1,r_2,\cdots,r_k\end{array}\bigg]_{p,q}
		+q^{x-s_1}\bigg[\begin{array}{c}x-1\\r_1-1,r_2,\cdots,r_k\end{array}\bigg]_{p,q}\cr
		&&\qquad\qquad\qquad+q^{x-s_2}\bigg[\begin{array}{c}x-1\\r_1,r_2-1,\cdots,r_k\end{array}\bigg]_{p,q}
		+\cdots+q^{x-s_k}\bigg[\begin{array}{c}x-1\\r_1,r_2,\cdots,r_{k-1}\end{array}\bigg]_{p,q},
	\end{eqnarray*}
\end{small}
where $r_j\in\mathbb{N}$ and $j\in\{1,2,\cdots,k\},$ with $m_j=\sum_{i=j}^kr_i$ and $s_j=\sum_{i=1}^jr_i.$

For $n$ a positive integers, $x,p,$ and $q$ real numbers, the following relation holds:
\begin{small}
	\begin{eqnarray*}
		\prod_{j=1}^{k}\big(1\oplus x_j\big)^{n}_{p,q}=\sum\bigg[\begin{array}{c}n\\r_1,r_2,\cdots,r_k\end{array}\bigg]_{p,q}\prod_{j=1}^{k}x^{r_j}_jp^{n-r_j\choose 2}q^{r_j\choose 2}\Big(p^{n-s_{j-1}}\oplus x_jq^{n-s_{j-1}}\Big)^{s_{j-1}}_{p,q},
	\end{eqnarray*}
\end{small}
where  $r_j\in\{0,\cdots,n\},$ $j\in\{1,\cdots,k\},$  with $\sum_{i=1}^{k}r_i\leq n$ and $s_j=\sum_{i=1}^{j}r_i,$ $s_0=0.$

Furthermore, for  $n$ be a positive integers, we have:
\begin{small}
	\begin{eqnarray*}
		\prod_{j=1}^{k}\big(1\oplus x_j\big)^{n}_{p,q}=\sum\bigg[\begin{array}{c}n+s_{k}-1\\r_1,r_2,\cdots,r_k\end{array}\bigg]_{p,q}\prod_{j=1}^{k}\frac{x^{r_j}_jp^{n-r_j\choose 2}q^{r_j\choose 2}}{ \big(p^{n}\oplus x_jq^{n}\big)^{s_k-s_{j-1}}_{p,q}}.
	\end{eqnarray*}
	Equivalently,
	\begin{eqnarray*}
		\prod_{j=1}^{k}\big(1\oplus x_j\big)^{n}_{p,q}&=&\sum_{r_j\in\mathbb{N}}\bigg[\begin{array}{c}n+s_{k}-1\\r_1,r_2,\cdots,r_k\end{array}\bigg]_{p,q}\prod_{j=1}^{k}\frac{x^{n+s_k-s_{j-1}}_jp^{n-r_j\choose 2}q^{{n+s_k-s_{j-1}\choose 2}+r_j}}{ \big(p^{n}\oplus x_jq^{n}\big)^{s_k-s_{j-1}}_{p,q}},
	\end{eqnarray*} 
	where $j\in\{1,2,\cdots,k\},$ with $s_j=\displaystyle\sum_{i=1}^{j}r_j,\quad s_0=0.$
\end{small}

Let $x_j, j\in\{1,2,\cdots, k+1\}, p,$ and $q$ real numbers. For $n$ positive integer, the following result holds. 
\begin{small}
	\begin{eqnarray*}
		\big(1 \ominus \Lambda_k\big)^{n}_{p,q}= \sum_{r_j=0}^{n}\bigg[\begin{array}{c}n\\r_1,r_2,\cdots,r_k\end{array}\bigg]_{p,q}\prod_{j=1}^{k}x^{n-s_j}_j\big(1\ominus x_j\big)^{r_j}_{p,q}\big(1\ominus x_{k+1}\big)^{n-s_k}_{p,q},
	\end{eqnarray*}
\end{small}
where $r_j\in\{0,\cdots,n\},$ $j\in\{1,\cdots,k\},$  with $\sum_{i=1}^{k}r_i\leq n$ and $s_j=\sum_{i=1}^{j}r_i,$ $s_0=0,$ $\Lambda_k=\prod_{j=1}^{k+1}x_j.$

For $n$  a positive integer, we have: 
\begin{eqnarray*}
	\sum_{r_j=0}^{n}\bigg[\begin{array}{c}n\\r_1,r_2,\cdots,r_k\end{array}\bigg]_{p,q}\prod_{j=1}^{k}\,x^{n-s_j}_j\big(1\ominus x_j\big)^{r_j}_{p,q}=p^{\frac{s_k(1+s_k-2n)}{2}}
\end{eqnarray*}
and
\begin{eqnarray*}
	\sum_{r_j=0}^{n}\bigg[\begin{array}{c}n\\r_1,r_2,\cdots,r_k\end{array}\bigg]_{p,q}\prod_{j=1}^{k}\,x^{r_j}_j\big(1\ominus x_j\big)^{n-s_j}_{p,q}=p^{\frac{s_k(1+s_k-2n)}{2}},
\end{eqnarray*}
where $j\in\{1,\cdots,k\},$  with $\displaystyle \sum_{i=1}^{k}r_j\leq n$ and $s_j=\sum_{i=1}^{j}r_i,$ $s_0=0.$

The $(p,q)-$ deformed of the multinomial formula given by {\bf Gasper and Rahman} \cite{GR} can be determined as follows:
\begin{small}
	\begin{eqnarray*}
		\big(1\ominus \Lambda_k\big)^{n}_{p,q}&=& \sum_{r_j=0}^{n}\bigg[\begin{array}{c}n\\r_1,r_2,\cdots,r_k\end{array}\bigg]_{p,q}\prod_{j=1}^{k}\,x^{s_j}_j\big(1\ominus x_{j-1}\big)^{n}_{p,q}\big(1\ominus x_{k}\big)^{n-s_k}_{p,q},
	\end{eqnarray*}
\end{small}
where $j\in\{1,2,\cdots, k\},$ with $\displaystyle \sum_{i=1}^{k}r_j\leq n$ and $s_j=\sum_{i=1}^{j}r_j.$
\begin{itemize}
	\item[(a)] 
	The probability function of the $(p,q)-$ deformed multinomial distribution of the first kind with parameters $n, \big(\theta_1,\theta_2,\cdots,\theta_k\big), p$ and $q$  
	is presented by:
	\begin{eqnarray*}
		P\big(Y_1=y_1,\cdots,Y_k=y_k\big)=\genfrac{[}{]}{0pt}{}{n}{y_1,y_2,\cdots,y_k}_{p,q}\prod_{j=1}^{k}\frac{\theta^{y_j}_j\,p^{n-x_j\choose 2}\,q^{y_j\choose 2}}{\big(1\oplus \theta_j\big)^{n-s_{j-1}}_{p,q}},
	\end{eqnarray*}
	and	their recursion relations as:
	\begin{eqnarray*}
		P_{y+1}= \Big[n-\sum_{j=1}^{k}y_j\Big]_{k,p,q}\prod_{j=1}^{k}\frac{\theta_j\,p^{n-y_j}q^{y_j}P_{y}}{[y_{j}+1]_{p,q}\big(1\oplus \theta_j\big)_{p,q}},\,\mbox{with}\, P_0= \prod_{j=1}^{k}\frac{p^{n\choose 2}}{\big(1\oplus\theta_j\big)^{n}_{p,q}},
	\end{eqnarray*} 
	where $y_j\in\{0,1,\cdots,n\}, \displaystyle\sum_{j=1}^{k}y_j\leq n, s_j=\displaystyle\sum_{i=1}^{j}y_j, 0<\theta_j<1$, and $j\in\{1,2,\cdots,k\}.$
	\item[(b)]
	The probability function of the negative $(p,q)-$ deformed multinomial distribution of the first kind with parameters $n, \big(\theta_1,\theta_2,\cdots,\theta_k\big), p$ and $q$ is given as follows:
	\begin{small}
		\begin{eqnarray*}\label{nrmfk}
			P\big(T_1=t_1,\cdots,T_k=t_k\big)=\genfrac{[}{]}{0pt}{}{n+s_k-1}{t_1,t_2,\cdots,t_k}_{p,q}\prod_{j=1}^{k}\frac{\theta^{u_j}_j\,p^{n-u_j\choose 2}\,q^{u_j\choose 2}}{\big(1\oplus \theta_j\big)^{n+s_k-s_{j-1}}_{p,q}},
		\end{eqnarray*}
	and their recurrence relation by:
	\begin{eqnarray*}
		P_{t+1}= \Big[n-\sum_{j=1}^{k}t_j\Big]_{k,p,q}\prod_{j=1}^{k}\frac{\theta_j\,p^{n-t_j}\,q^{t_j}P_{t}}{[t_{j}+1]_{p,q}\big(1\ominus \theta_j\big)_{p,q}},\,\mbox{with}\, P_0= \prod_{j=1}^{k}\frac{p^{n\choose 2}}{\big(1\oplus\theta_j\big)^{n}_{p,q}},
	\end{eqnarray*}
	\end{small}
	where $t_j\in\mathbb{N}, s_j=\displaystyle\sum_{i=1}^{j}t_j, 0<\theta_j<1$, and $j\in\{1,2,\cdots,k\}.$
	\item[(c)] The probability function of the $(p,q)-$ deformed multinomial distribution of the second kind with parameters $n, \big(\theta_1,\theta_2,\cdots,\theta_k\big), p$ and $q$ is determined by:
	\begin{eqnarray*}\label{rmsk}
		P\big(X_1=x_1,\cdots,X_k=x_k\big)=\genfrac{[}{]}{0pt}{}{n}{x_1,x_2,\cdots,x_k}_{p,q}\prod_{j=1}^{k}\theta^{x_j}_j\big(1\ominus \theta_j\big)^{n-s_{j}}_{p,q}
	\end{eqnarray*}
and the recurrence relation 
	\begin{eqnarray*}
		P_{x+1}= \Big[n-\sum_{j=1}^{k}x_j\Big]_{k,p,q}\prod_{j=1}^{k}\frac{\theta_j\big(1\ominus \theta_j\big)_{p,q}}{[x_{j}+1]_{p,q}}P_{x},\quad\mbox{with}\quad P_0= \prod_{j=1}^{k}\big(1\ominus\theta_j\big)^{n}_{p,q}.
	\end{eqnarray*}	where $x_j\in\{0,1,\cdots,n\}, \displaystyle\sum_{j=1}^{k}x_j\leq n, s_j=\displaystyle\sum_{i=1}^{j}x_j.$

	{Another $(p,q)-$ deformed multinomial distribution of the second kind}
		\begin{eqnarray*}
		P\big(Y_1=y_1,\cdots,Y_k=y_k\big)=\bigg[\begin{array}{c} n  \\ y_1,y_2,\cdots,y_k\end{array} \bigg]_{p,q}\prod_{j=1}^{k}\theta^{n-s_j}_j\big(1\ominus \theta_j\big)^{y_j}_{p,q},
		\end{eqnarray*}
	where $y_j\in\{0,1,\cdots,n\}, \displaystyle\sum_{j=1}^{k}y_j\leq n, s_j=\displaystyle\sum_{i=1}^{j}y_j, 0<\theta_j<1$, and $j\in\{1,2,\cdots,k\}.$
	\item[(d)] The probability function of the negative $(p,q)-$ deformed multinomial distribution of the second kind with parameters $n, \big(\theta_1,\theta_2,\cdots,\theta_k\big), p$ and $q$ is furnished by:
	\begin{small}
		\begin{eqnarray*}\label{nrmsk}
			P\big(W_1=w_1,\cdots,W_k=w_k\big)=\bigg[\begin{array}{c} n+s_k-1  \\ w_1,w_2,\cdots,w_k\end{array} \bigg]_{p,q}\prod_{j=1}^{k}{\theta^{w_j}_j \big(1\ominus \theta_j\big)^{n+s_k-s_{j}}_{p,q}}.
		\end{eqnarray*}
	Furthermore, their recursion relations are given as follows:
	\begin{eqnarray*}
		P_{x+1}= \Big[n-\sum_{j=1}^{k}x_j\Big]_{k,p,q}\prod_{j=1}^{k}\frac{\theta_j\,\big(1\ominus \theta_j\big)_{p,q}}{[x_{j}+1]_{p,q}}P_{x},\,\mbox{with}\quad P_0= \prod_{j=1}^{k}\big(1\ominus\theta_j\big)^{n}_{p,q}.
	\end{eqnarray*}
	\end{small}
	where $w_j\in\mathbb{N}, s_j=\displaystyle\sum_{i=1}^{j}w_j, 0<\theta_j<1$, and $j\in\{1,2,\cdots,k\}.$
	\end{itemize}
\item[(ii)] Putting $\mathcal{R}(x,y)=\frac{1-xy}{(p^{-1}-q)x},$ we obtain the multinomial distribution and properties corresponding to the {\bf Chakrabarty and Jagannathan} algebra \cite{CJ}: the {$(p^{-1},q)$- deformed multinomial coefficient}
\begin{eqnarray*}
	\bigg[\begin{array}{c}x\\r_1,r_2,\cdots,r_k\end{array}\bigg]_{p^{-1},q}=\frac{[x]_{r_1+r_2+\cdots+r_k,p^{-1},q}}{ [r_1]_{p^{-1},q}![r_2]_{p^{-1},q}!\cdots[r_k]_{p^{-1},q}!}
\end{eqnarray*}
satisfies the recursion relation:
\begin{small}
	\begin{eqnarray*}
		&&\bigg[\begin{array}{c}x\\r_1,r_2,\cdots,r_k\end{array}\bigg]_{p^{-1},q}=p^{-s_k}\,\bigg[\begin{array}{c}x-1\\r_1,r_2,\cdots,r_k\end{array}\bigg]_{p^{-1},q}
		+q^{x-m_1}\bigg[\begin{array}{c}x-1\\r_1-1,r_2,\cdots,r_k\end{array}\bigg]_{p^{-1},q}\cr
		&&\qquad\qquad+q^{x-m_2}\bigg[\begin{array}{c}x-1\\r_1,r_2-1,\cdots,r_k\end{array}\bigg]_{p^{-1},q}
		+\cdots+q^{x-m_k}\bigg[\begin{array}{c}x-1\\r_1,r_2,\cdots,r_{k-1}\end{array}\bigg]_{p^{-1},q}
	\end{eqnarray*}
	and alternatively,
	\begin{eqnarray*}
		&&\bigg[\begin{array}{c}x\\r_1,r_2,\cdots,r_k\end{array}\bigg]_{p^{-1},q}=q^{s_k}\bigg[\begin{array}{c}x-1\\r_1,r_2,\cdots,r_k\end{array}\bigg]_{p^{-1},q}
		+p^{-x+m_1}\bigg[\begin{array}{c}x-1\\r_1-1,r_2,\cdots,r_k\end{array}\bigg]_{p^{-1},q}\cr
		&&\qquad\qquad+p^{-x+m_2}q^{s_1}\bigg[\begin{array}{c}x-1\\r_1,r_2-1,\cdots,r_k\end{array}\bigg]_{p^{-1},q}
		+\cdots+p^{-x+m_k}q^{s_{k-1}}\bigg[\begin{array}{c}x-1\\r_1,r_2,\cdots,r_{k-1}\end{array}\bigg]_{p^{-1},q}.
	\end{eqnarray*}
\end{small}
Moreover, 
the $(p,q^{-1})$- deformed multinomial coefficient provided by
\begin{eqnarray*}
	\bigg[\begin{array}{c}x\\r_1,r_2,\cdots,r_k\end{array}\bigg]_{p,q^{-1}}
	&=&(p^{-1}q)^{-\displaystyle\sum_{j=1}^k r_j(x-m_j)}\bigg[\begin{array}{c}x\\r_1,r_2,\cdots,r_k\end{array}\bigg]_{p^{-1},q}\nonumber\\
	&=&(p^{-1}q)^{-\displaystyle\sum_{j=1}^k r_j(x-s_j)}\bigg[\begin{array}{c}x\\r_1,r_2,\cdots,r_k\end{array}\bigg]_{p^{-1},q}
\end{eqnarray*}
obey the recursion relation: 
\begin{small}
	\begin{eqnarray*}
		&&\bigg[\begin{array}{c}x\\r_1,r_2,\cdots,r_k\end{array}\bigg]_{p^{-1},q}=q^{m_1}\bigg[\begin{array}{c}x-1\\r_1,r_2,\cdots,r_k\end{array}\bigg]_{p^{-1},q}
		+q^{m_2}\bigg[\begin{array}{c}x-1\\r_1-1,r_2,\cdots,r_k\end{array}\bigg]_{p^{-1},q}\cr&
		&\qquad\qquad+q^{m_3}\bigg[\begin{array}{c}x-1\\r_1,r_2-1,\cdots,r_k\end{array}\bigg]_{p^{-1},q}
		+\cdots+ p^{-x}\bigg[\begin{array}{c}x-1\\r_1,r_2,\cdots,r_{k-1}\end{array}\bigg]_{p^{-1},q}.
	\end{eqnarray*}
\end{small}
and
\begin{small}
	\begin{eqnarray*} 
		&&\bigg[\begin{array}{c}x\\r_1,r_2,\cdots,r_k\end{array}\bigg]_{p^{-1},q}=p^{-x}\bigg[\begin{array}{c}x-1\\r_1,r_2,\cdots,r_k\end{array}\bigg]_{p^{-1},q}
		+q^{x-s_1}\bigg[\begin{array}{c}x-1\\r_1-1,r_2,\cdots,r_k\end{array}\bigg]_{p^{-1},q}\cr
		&&\qquad\qquad+q^{x-s_2}\bigg[\begin{array}{c}x-1\\r_1,r_2-1,\cdots,r_k\end{array}\bigg]_{p^{-1},q}
		+\cdots+q^{x-s_k}\bigg[\begin{array}{c}x-1\\r_1,r_2,\cdots,r_{k-1}\end{array}\bigg]_{p^{-1},q},
	\end{eqnarray*}
\end{small}
where $r_j\in\mathbb{N}$ and $j\in\{1,2,\cdots,k\},$ with $m_j=\sum_{i=j}^kr_i$ and $s_j=\sum_{i=1}^jr_i.$

For $n$ a positive integers, $x,$ and $q$ real numbers, the following relation holds:
\begin{small}
	\begin{eqnarray*}
		\prod_{j=1}^{k}\big(1\oplus x_j\big)^{n}_{p^{-1},q}=\sum\bigg[\begin{array}{c}n\\r_1,r_2,\cdots,r_k\end{array}\bigg]_{p^{-1},q}\prod_{j=1}^{k}x^{r_j}_jp^{-{n-r_j\choose 2}}q^{r_j\choose 2}\Big(p^{-n+s_{j-1}}\oplus x_jq^{n-s_{j-1}}\Big)^{s_{j-1}}_{p^{-1},q},
	\end{eqnarray*}
\end{small}
where  $r_j\in\{0,\cdots,n\},$ $j\in\{1,\cdots,k\},$  with $\sum_{i=1}^{k}r_i\leq n$ and $s_j=\sum_{i=1}^{j}r_i,$ $s_0=0.$

Furthermore, for  $n$ be a positive integers, we have:
\begin{small}
	\begin{eqnarray*}
		\prod_{j=1}^{k}\big(1\oplus x_j\big)^{n}_{p^{-1},q}=\sum\bigg[\begin{array}{c}n+s_{k}-1\\r_1,r_2,\cdots,r_k\end{array}\bigg]_{p^{-1},q}\prod_{j=1}^{k}\frac{x^{r_j}_jp^{-{n-r_j\choose 2}}q^{r_j\choose 2}}{ \big(p^{n}\oplus x_jq^{n}\big)^{s_k-s_{j-1}}_{p^{-1},q}}.
	\end{eqnarray*}
	Equivalently,
	\begin{eqnarray*}
		\prod_{j=1}^{k}\big(1\oplus x_j\big)^{n}_{p^{-1},q}&=&\sum_{r_j\in\mathbb{N}}\bigg[\begin{array}{c}n+s_{k}-1\\r_1,r_2,\cdots,r_k\end{array}\bigg]_{p^{-1},q}\prod_{j=1}^{k}\frac{x^{n+s_k-s_{j-1}}_jp^{-{n-r_j\choose 2}}q^{{n+s_k-s_{j-1}\choose 2}+r_j}}{ \big(p^{n}\oplus x_jq^{n}\big)^{s_k-s_{j-1}}_{p^{-1},q}},
	\end{eqnarray*} 
	where $j\in\{1,2,\cdots,k\},$ with $s_j=\displaystyle\sum_{i=1}^{j}r_j,\quad s_0=0.$
\end{small}

Let $x_j, j\in\{1,2,\cdots, k+1\}, p,$ and $q$ real numbers. For $n$ positive integer, the following result holds. 
\begin{small}
	\begin{eqnarray*}
		\big(1 \ominus \Lambda_k\big)^{n}_{p^{-1},q}= \sum_{r_j=0}^{n}\bigg[\begin{array}{c}n\\r_1,r_2,\cdots,r_k\end{array}\bigg]_{p^{-1},q}\prod_{j=1}^{k}x^{n-s_j}_j\big(1\ominus x_j\big)^{r_j}_{p^{-1},q}\big(1\ominus x_{k+1}\big)^{n-s_k}_{p^{-1},q},
	\end{eqnarray*}
\end{small}
where $r_j\in\{0,\cdots,n\},$ $j\in\{1,\cdots,k\},$  with $\sum_{i=1}^{k}r_i\leq n$ and $s_j=\sum_{i=1}^{j}r_i,$ $s_0=0,$ $\Lambda_k=\prod_{j=1}^{k+1}x_j.$

For $n$  a positive integer, we have: 
\begin{eqnarray*}
	\sum_{r_j=0}^{n}\bigg[\begin{array}{c}n\\r_1,r_2,\cdots,r_k\end{array}\bigg]_{p^{-1},q}\prod_{j=1}^{k}\,x^{n-s_j}_j\big(1\ominus x_j\big)^{r_j}_{p^{-1},q}=p^{-\frac{s_k(1+s_k-2n)}{2}}
\end{eqnarray*}
and
\begin{eqnarray*}
	\sum_{r_j=0}^{n}\bigg[\begin{array}{c}n\\r_1,r_2,\cdots,r_k\end{array}\bigg]_{p^{-1},q}\prod_{j=1}^{k}\,x^{r_j}_j\big(1\ominus x_j\big)^{n-s_j}_{p^{-1},q}=p^{-\frac{s_k(1+s_k-2n)}{2}},
\end{eqnarray*}
where $j\in\{1,\cdots,k\},$  with $\displaystyle \sum_{i=1}^{k}r_j\leq n$ and $s_j=\sum_{i=1}^{j}r_i,$ $s_0=0.$

The $(p^{-1},q)-$ deformed of the multinomial formula given by {\bf Gasper and Rahman} \cite{GR} can be determined as follows:
\begin{small}
	\begin{eqnarray*}
		\big(1\ominus \Lambda_k\big)^{n}_{p^{-1},q}&=& \sum_{r_j=0}^{n}\bigg[\begin{array}{c}n\\r_1,r_2,\cdots,r_k\end{array}\bigg]_{p^{-1},q}\prod_{j=1}^{k}\,x^{s_j}_j\big(1\ominus x_{j-1}\big)^{n}_{p^{-1},q}\big(1\ominus x_{k}\big)^{n-s_k}_{p^{-1},q},
	\end{eqnarray*}
\end{small}
where $j\in\{1,2,\cdots, k\},$ with $\displaystyle \sum_{i=1}^{k}r_j\leq n$ and $s_j=\sum_{i=1}^{j}r_j.$
\begin{itemize}
	\item[(a)] 
	The probability function of the $(p^{-1},q)-$ deformed multinomial distribution of the first kind with parameters $n, \big(\theta_1,\theta_2,\cdots,\theta_k\big), p$ and $q$  
	is presented by:
	\begin{eqnarray*}
		P\big(Y_1=y_1,\cdots,Y_k=y_k\big)=\genfrac{[}{]}{0pt}{}{n}{y_1,y_2,\cdots,y_k}_{p^{-1},q}\prod_{j=1}^{k}\frac{\theta^{y_j}_j\,p^{-{n-x_j\choose 2}}\,q^{y_j\choose 2}}{\big(1\oplus \theta_j\big)^{n-s_{j-1}}_{p^{-1},q}},
	\end{eqnarray*}
	and	their recursion relations as:
	\begin{small}
	\begin{eqnarray*}
		P_{y+1}= \Big[n-\sum_{j=1}^{k}y_j\Big]_{k,p^{-1},q}\prod_{j=1}^{k}\frac{\theta_j\,p^{-n+y_j}q^{y_j}P_{y}}{[y_{j}+1]_{p^{-1},q}\big(1\oplus \theta_j\big)_{p^{-1},q}},\,\mbox{with}\, P_0= \prod_{j=1}^{k}\frac{p^{-{n\choose 2}}}{\big(1\oplus\theta_j\big)^{n}_{p^{-1},q}},
	\end{eqnarray*} 
\end{small}
	where $y_j\in\{0,1,\cdots,n\}, \displaystyle\sum_{j=1}^{k}y_j\leq n, s_j=\displaystyle\sum_{i=1}^{j}y_j, 0<\theta_j<1$, and $j\in\{1,2,\cdots,k\}.$
	\item[(b)]
	The probability function of the negative $(p^{-1},q)-$ deformed multinomial distribution of the first kind with parameters $n, \big(\theta_1,\theta_2,\cdots,\theta_k\big), p$ and $q$ is given as follows:
	\begin{small}
		\begin{eqnarray*}
			P\big(T_1=t_1,\cdots,T_k=t_k\big)=\genfrac{[}{]}{0pt}{}{n+s_k-1}{t_1,t_2,\cdots,t_k}_{p^{-1},q}\prod_{j=1}^{k}\frac{\theta^{u_j}_j\,p^{-{n-u_j\choose 2}}\,q^{u_j\choose 2}}{\big(1\oplus \theta_j\big)^{n+s_k-s_{j-1}}_{p^{-1},q}},
		\end{eqnarray*}
		and their recurrence relation by:
		\begin{eqnarray*}
			P_{t+1}= \Big[n-\sum_{j=1}^{k}t_j\Big]_{k,p^{-1},q}\prod_{j=1}^{k}\frac{\theta_j\,p^{-n+t_j}\,q^{t_j}P_{t}}{[t_{j}+1]_{p^{-1},q}\big(1\ominus \theta_j\big)_{p^{-1},q}},\,\mbox{with}\, P_0= \prod_{j=1}^{k}\frac{p^{-{n\choose 2}}}{\big(1\oplus\theta_j\big)^{n}_{p^{-1},q}},
		\end{eqnarray*}
	\end{small}
	where $t_j\in\mathbb{N}, s_j=\displaystyle\sum_{i=1}^{j}t_j, 0<\theta_j<1$, and $j\in\{1,2,\cdots,k\}.$
	\item[(c)] The probability function of the $(p^{-1},q)-$ deformed multinomial distribution of the second kind with parameters $n, \big(\theta_1,\theta_2,\cdots,\theta_k\big), p$ and $q$ is determined by:
	\begin{eqnarray*}
		P\big(X_1=x_1,\cdots,X_k=x_k\big)=\genfrac{[}{]}{0pt}{}{n}{x_1,x_2,\cdots,x_k}_{p^{-1},q}\prod_{j=1}^{k}\theta^{x_j}_j\big(1\ominus \theta_j\big)^{n-s_{j}}_{p^{-1},q}
	\end{eqnarray*}
	and the recurrence relation 
	\begin{eqnarray*}
		P_{x+1}= \Big[n-\sum_{j=1}^{k}x_j\Big]_{k,p^{-1},q}\prod_{j=1}^{k}\frac{\theta_j\big(1\ominus \theta_j\big)_{p^{-1},q}}{[x_{j}+1]_{p^{-1},q}}P_{x},\quad\mbox{with}\quad P_0= \prod_{j=1}^{k}\big(1\ominus\theta_j\big)^{n}_{p^{-1},q}.
	\end{eqnarray*}	where $x_j\in\{0,1,\cdots,n\}, \displaystyle\sum_{j=1}^{k}x_j\leq n, s_j=\displaystyle\sum_{i=1}^{j}x_j.$
	
	{Another $(p^{-1},q)-$ deformed multinomial distribution of the second kind}
	\begin{eqnarray*}
		P\big(Y_1=y_1,\cdots,Y_k=y_k\big)=\bigg[\begin{array}{c} n  \\ y_1,y_2,\cdots,y_k\end{array} \bigg]_{p^{-1},q}\prod_{j=1}^{k}\theta^{n-s_j}_j\big(1\ominus \theta_j\big)^{y_j}_{p^{-1},q},
	\end{eqnarray*}
	where $y_j\in\{0,1,\cdots,n\}, \displaystyle\sum_{j=1}^{k}y_j\leq n, s_j=\displaystyle\sum_{i=1}^{j}y_j, 0<\theta_j<1$, and $j\in\{1,2,\cdots,k\}.$
	\item[(d)] The probability function of the negative $(p^{-1},q)-$ deformed multinomial distribution of the second kind with parameters $n, \big(\theta_1,\theta_2,\cdots,\theta_k\big), p$ and $q$ is furnished by:
	\begin{small}
		\begin{eqnarray*}
			P\big(W_1=w_1,\cdots,W_k=w_k\big)=\bigg[\begin{array}{c} n+s_k-1  \\ w_1,w_2,\cdots,w_k\end{array} \bigg]_{p^{-1},q}\prod_{j=1}^{k}{\theta^{w_j}_j \big(1\ominus \theta_j\big)^{n+s_k-s_{j}}_{p^{-1},q}}.
		\end{eqnarray*}
		Furthermore, their recursion relations are given as follows:
		\begin{eqnarray*}
			P_{x+1}= \Big[n-\sum_{j=1}^{k}x_j\Big]_{k,p^{-1},q}\prod_{j=1}^{k}\frac{\theta_j\,\big(1\ominus \theta_j\big)_{p^{-1},q}}{[x_{j}+1]_{p^{-1},q}}P_{x},\,\mbox{with}\quad P_0= \prod_{j=1}^{k}\big(1\ominus\theta_j\big)^{n}_{p^{-1},q}.
		\end{eqnarray*}
	\end{small}
	where $w_j\in\mathbb{N}, s_j=\displaystyle\sum_{i=1}^{j}w_j, 0<\theta_j<1$, and $j\in\{1,2,\cdots,k\}.$
\end{itemize}
\item[(iii)] The multinomial distribution and properties associated to the {\bf Hounkonnou-Ngompe generalized $q-$ Quesne} algebra \cite{HN} can be deduced by putting $\mathcal{R}(x,y)=\frac{xy-1}{(q-p^{-1})y}:$ the corresponding {  multinomial coefficient}
\begin{eqnarray*}
	\bigg[\begin{array}{c}x\\r_1,r_2,\cdots,r_k\end{array}\bigg]^Q_{p,q}=\frac{[x]^Q_{r_1+r_2+\cdots+r_k,p,q}}{ [r_1]^Q_{p,q}![r_2]^Q_{p,q}!\cdots[r_k]^Q_{p,q}!}
\end{eqnarray*}
satisfies the recursion relation:
\begin{small}
	\begin{eqnarray*}
		&&\bigg[\begin{array}{c}x\\r_1,r_2,\cdots,r_k\end{array}\bigg]^Q_{p,q}=p^{s_k}\,\bigg[\begin{array}{c}x-1\\r_1,r_2,\cdots,r_k\end{array}\bigg]^Q_{p,q}
		+q^{-x+m_1}\bigg[\begin{array}{c}x-1\\r_1-1,r_2,\cdots,r_k\end{array}\bigg]^Q_{p,q}\cr
		&&\qquad+q^{-x+m_2}\bigg[\begin{array}{c}x-1\\r_1,r_2-1,\cdots,r_k\end{array}\bigg]^Q_{p,q}
		+\cdots+q^{-x+m_k}\bigg[\begin{array}{c}x-1\\r_1,r_2,\cdots,r_{k-1}\end{array}\bigg]^Q_{p,q}
	\end{eqnarray*}
	and alternatively,
	\begin{eqnarray*}
		&&\bigg[\begin{array}{c}x\\r_1,r_2,\cdots,r_k\end{array}\bigg]^Q_{p,q}=q^{-s_k}\bigg[\begin{array}{c}x-1\\r_1,r_2,\cdots,r_k\end{array}\bigg]^Q_{p,q}
		+p^{x-m_1}\bigg[\begin{array}{c}x-1\\r_1-1,r_2,\cdots,r_k\end{array}\bigg]^Q_{p,q}\cr
		&&\qquad+p^{x-m_2}q^{-s_1}\bigg[\begin{array}{c}x-1\\r_1,r_2-1,\cdots,r_k\end{array}\bigg]^Q_{p,q}
		+\cdots+p^{x-m_k}q^{-s_{k-1}}\bigg[\begin{array}{c}x-1\\r_1,r_2,\cdots,r_{k-1}\end{array}\bigg]^Q_{p,q}.
	\end{eqnarray*}
\end{small}
Moreover, 
the $(p^{-1},q)$- deformed multinomial coefficient provided by
\begin{eqnarray*}
	\bigg[\begin{array}{c}x\\r_1,r_2,\cdots,r_k\end{array}\bigg]^Q_{p^{-1},q^{-1}}
	&=&(pq^{-1})^{-\displaystyle\sum_{j=1}^k r_j(x-m_j)}\bigg[\begin{array}{c}x\\r_1,r_2,\cdots,r_k\end{array}\bigg]^Q_{p,q}\nonumber\\
	&=&(pq^{-1})^{-\displaystyle\sum_{j=1}^k r_j(x-s_j)}\bigg[\begin{array}{c}x\\r_1,r_2,\cdots,r_k\end{array}\bigg]^Q_{p,q}
\end{eqnarray*}
obey the recursion relation: 
\begin{small}
	\begin{eqnarray*}
		&&\bigg[\begin{array}{c}x\\r_1,r_2,\cdots,r_k\end{array}\bigg]^Q_{p,q}=q^{-m_1}\bigg[\begin{array}{c}x-1\\r_1,r_2,\cdots,r_k\end{array}\bigg]^Q_{p,q}
		+q^{-m_2}\bigg[\begin{array}{c}x-1\\r_1-1,r_2,\cdots,r_k\end{array}\bigg]^Q_{p,q}\cr&
		&\qquad\qquad+q^{-m_3}\bigg[\begin{array}{c}x-1\\r_1,r_2-1,\cdots,r_k\end{array}\bigg]^Q_{p,q}
		+\cdots+ p^x\bigg[\begin{array}{c}x-1\\r_1,r_2,\cdots,r_{k-1}\end{array}\bigg]^Q_{p,q}
	\end{eqnarray*}
\end{small}
and
\begin{small}
	\begin{eqnarray*} 
		&&\bigg[\begin{array}{c}x\\r_1,r_2,\cdots,r_k\end{array}\bigg]^Q_{p,q}=p^x\bigg[\begin{array}{c}x-1\\r_1,r_2,\cdots,r_k\end{array}\bigg]^Q_{p,q}
		+q^{-x+s_1}\bigg[\begin{array}{c}x-1\\r_1-1,r_2,\cdots,r_k\end{array}\bigg]^Q_{p,q}\cr
		&&\qquad\qquad+q^{-x+s_2}\bigg[\begin{array}{c}x-1\\r_1,r_2-1,\cdots,r_k\end{array}\bigg]^Q_{p,q}
		+\cdots+q^{-x+s_k}\bigg[\begin{array}{c}x-1\\r_1,r_2,\cdots,r_{k-1}\end{array}\bigg]^Q_{p,q},
	\end{eqnarray*}
\end{small}
where $r_j\in\mathbb{N}$ and $j\in\{1,2,\cdots,k\},$ with $m_j=\sum_{i=j}^kr_i$ and $s_j=\sum_{i=1}^jr_i.$

For $n$ a positive integers, $x,p,$ and $q$ real numbers, the following relation holds:
\begin{small}
	\begin{eqnarray*}
		\prod_{j=1}^{k}\big(1\oplus x_j\big)^{n}_{p,q^{-1}}=\sum\bigg[\begin{array}{c}n\\r_1,r_2,\cdots,r_k\end{array}\bigg]^Q_{p,q}\prod_{j=1}^{k}x^{r_j}_jp^{n-r_j\choose 2}q^{-{r_j\choose 2}}\Big(p^{n-s_{j-1}}\oplus x_jq^{-n+s_{j-1}}\Big)^{s_{j-1}}_{p,q^{-1}},
	\end{eqnarray*}
\end{small}
where  $r_j\in\{0,\cdots,n\},$ $j\in\{1,\cdots,k\},$  with $\sum_{i=1}^{k}r_i\leq n$ and $s_j=\sum_{i=1}^{j}r_i,$ $s_0=0.$

Furthermore, for  $n$ be a positive integers, we have:
\begin{small}
	\begin{eqnarray*}
		\prod_{j=1}^{k}\big(1\oplus x_j\big)^{n}_{p,q^{-1}}=\sum\bigg[\begin{array}{c}n+s_{k}-1\\r_1,r_2,\cdots,r_k\end{array}\bigg]^Q_{p,q}\prod_{j=1}^{k}\frac{x^{r_j}_jp^{n-r_j\choose 2}q^{-{r_j\choose 2}}}{ \big(p^{n}\oplus x_jq^{n}\big)^{s_k-s_{j-1}}_{p,q^{-1}}}.
	\end{eqnarray*}
	Equivalently,
	\begin{eqnarray*}
		\prod_{j=1}^{k}\big(1\oplus x_j\big)^{n}_{p,q^{-1}}&=&\sum_{r_j\in\mathbb{N}}\bigg[\begin{array}{c}n+s_{k}-1\\r_1,r_2,\cdots,r_k\end{array}\bigg]^Q_{p,q}\prod_{j=1}^{k}\frac{x^{n+s_k-s_{j-1}}_jp^{n-r_j\choose 2}q^{-{n+s_k-s_{j-1}\choose 2}-r_j}}{ \big(p^{n}\oplus x_jq^{n}\big)^{s_k-s_{j-1}}_{p,q^{-1}}},
	\end{eqnarray*} 
	where $j\in\{1,2,\cdots,k\},$ with $s_j=\displaystyle\sum_{i=1}^{j}r_j,\quad s_0=0.$
\end{small}

Let $x_j, j\in\{1,2,\cdots, k+1\}, p,$ and $q$ real numbers. For $n$ positive integer, the following result holds. 
\begin{small}
	\begin{eqnarray*}
		\big(1 \ominus \Lambda_k\big)^{n}_{p,q^{-1}}= \sum_{r_j=0}^{n}\bigg[\begin{array}{c}n\\r_1,r_2,\cdots,r_k\end{array}\bigg]^Q_{p,q}\prod_{j=1}^{k}x^{n-s_j}_j\big(1\ominus x_j\big)^{r_j}_{p,q^{-1}}\big(1\ominus x_{k+1}\big)^{n-s_k}_{p,q^{-1}},
	\end{eqnarray*}
\end{small}
where $r_j\in\{0,\cdots,n\},$ $j\in\{1,\cdots,k\},$  with $\sum_{i=1}^{k}r_i\leq n$ and $s_j=\sum_{i=1}^{j}r_i,$ $s_0=0,$ $\Lambda_k=\prod_{j=1}^{k+1}x_j.$

For $n$  a positive integer, we have: 
\begin{eqnarray*}
	\sum_{r_j=0}^{n}\bigg[\begin{array}{c}n\\r_1,r_2,\cdots,r_k\end{array}\bigg]^Q_{p,q}\prod_{j=1}^{k}\,x^{n-s_j}_j\big(1\ominus x_j\big)^{r_j}_{p,q^{-1}}=p^{\frac{s_k(1+s_k-2n)}{2}}
\end{eqnarray*}
and
\begin{eqnarray*}
	\sum_{r_j=0}^{n}\bigg[\begin{array}{c}n\\r_1,r_2,\cdots,r_k\end{array}\bigg]^Q_{p,q}\prod_{j=1}^{k}\,x^{r_j}_j\big(1\ominus x_j\big)^{n-s_j}_{p,q^{-1}}=p^{\frac{s_k(1+s_k-2n)}{2}},
\end{eqnarray*}
where $j\in\{1,\cdots,k\},$  with $\displaystyle \sum_{i=1}^{k}r_j\leq n$ and $s_j=\sum_{i=1}^{j}r_i,$ $s_0=0.$

The associated multinomial formula given by {\bf Gasper and Rahman} \cite{GR} can be determined as follows:
\begin{small}
	\begin{eqnarray*}
		\big(1\ominus \Lambda_k\big)^{n}_{p,q^{-1}}&=& \sum_{r_j=0}^{n}\bigg[\begin{array}{c}n\\r_1,r_2,\cdots,r_k\end{array}\bigg]^Q_{p,q}\prod_{j=1}^{k}\,x^{s_j}_j\big(1\ominus x_{j-1}\big)^{n}_{p,q^{-1}}\big(1\ominus x_{k}\big)^{n-s_k}_{p,q^{-1}},
	\end{eqnarray*}
\end{small}
where $j\in\{1,2,\cdots, k\},$ with $\displaystyle \sum_{i=1}^{k}r_j\leq n$ and $s_j=\sum_{i=1}^{j}r_j.$
\begin{itemize}
	\item[(a)] 
	The probability function of the  deformed multinomial distribution of the first kind with parameters $n, \big(\theta_1,\theta_2,\cdots,\theta_k\big), p$ and $q$  
	is presented by:
	\begin{eqnarray*}
		P\big(Y_1=y_1,\cdots,Y_k=y_k\big)=\genfrac{[}{]}{0pt}{}{n}{y_1,y_2,\cdots,y_k}^Q_{p,q}\prod_{j=1}^{k}\frac{\theta^{y_j}_j\,p^{n-x_j\choose 2}\,q^{-{y_j\choose 2}}}{\big(1\oplus \theta_j\big)^{n-s_{j-1}}_{p,q^{-1}}},
	\end{eqnarray*}
	and	their recursion relations as:
	\begin{eqnarray*}
		P_{y+1}= \Big[n-\sum_{j=1}^{k}y_j\Big]^Q_{k,p,q}\prod_{j=1}^{k}\frac{\theta_j\,p^{n-y_j}q^{-y_j}P_{y}}{[y_{j}+1]^Q_{p,q}\big(1\oplus \theta_j\big)_{p,q^{-1}}},\,\mbox{with}\, P_0= \prod_{j=1}^{k}\frac{p^{n\choose 2}}{\big(1\oplus\theta_j\big)^{n}_{p,q^{-1}}},
	\end{eqnarray*} 
	where $y_j\in\{0,1,\cdots,n\}, \displaystyle\sum_{j=1}^{k}y_j\leq n, s_j=\displaystyle\sum_{i=1}^{j}y_j, 0<\theta_j<1$, and $j\in\{1,2,\cdots,k\}.$
	\item[(b)]
	The probability function of the negative  deformed multinomial distribution of the first kind with parameters $n, \big(\theta_1,\theta_2,\cdots,\theta_k\big), p$ and $q$ is given as follows:
	\begin{small}
		\begin{eqnarray*}
			P\big(T_1=t_1,\cdots,T_k=t_k\big)=\genfrac{[}{]}{0pt}{}{n+s_k-1}{t_1,t_2,\cdots,t_k}^Q_{p,q}\prod_{j=1}^{k}\frac{\theta^{u_j}_j\,p^{n-u_j\choose 2}\,q^{-{u_j\choose 2}}}{\big(1\oplus \theta_j\big)^{n+s_k-s_{j-1}}_{p,q^{-1}}},
		\end{eqnarray*}
		and their recurrence relation by:
		\begin{eqnarray*}
			P_{t+1}= \Big[n-\sum_{j=1}^{k}t_j\Big]^Q_{k,p,q}\prod_{j=1}^{k}\frac{\theta_j\,p^{n-t_j}\,q^{-t_j}P_{t}}{[t_{j}+1]^Q_{p,q}\big(1\ominus \theta_j\big)_{p,q^{-1}}},\,\mbox{with}\, P_0= \prod_{j=1}^{k}\frac{p^{n\choose 2}}{\big(1\oplus\theta_j\big)^{n}_{p,q^{-1}}},
		\end{eqnarray*}
	\end{small}
	where $t_j\in\mathbb{N}, s_j=\displaystyle\sum_{i=1}^{j}t_j, 0<\theta_j<1$, and $j\in\{1,2,\cdots,k\}.$
	\item[(c)] The probability function of the  deformed multinomial distribution of the second kind with parameters $n, \big(\theta_1,\theta_2,\cdots,\theta_k\big), p$ and $q$ is determined by:
	\begin{eqnarray*}
		P\big(X_1=x_1,\cdots,X_k=x_k\big)=\genfrac{[}{]}{0pt}{}{n}{x_1,x_2,\cdots,x_k}^Q_{p,q}\prod_{j=1}^{k}\theta^{x_j}_j\big(1\ominus \theta_j\big)^{n-s_{j}}_{p,q^{-1}}
	\end{eqnarray*}
	and the recurrence relation 
	\begin{eqnarray*}
		P_{x+1}= \Big[n-\sum_{j=1}^{k}x_j\Big]^Q_{k,p,q}\prod_{j=1}^{k}\frac{\theta_j\big(1\ominus \theta_j\big)_{p,q^{-1}}}{[x_{j}+1]^Q_{p,q}}P_{x},
	\end{eqnarray*}
where $$ P_0= \prod_{j=1}^{k}\big(1\ominus\theta_j\big)^{n}_{p,q^{-1}}$$
	and  $x_j\in\{0,1,\cdots,n\}, \displaystyle\sum_{j=1}^{k}x_j\leq n, s_j=\displaystyle\sum_{i=1}^{j}x_j.$
	
	Another deformed multinomial distribution of the second kind
	\begin{eqnarray*}
		P\big(Y_1=y_1,\cdots,Y_k=y_k\big)=\bigg[\begin{array}{c} n  \\ y_1,y_2,\cdots,y_k\end{array} \bigg]^Q_{p,q}\prod_{j=1}^{k}\theta^{n-s_j}_j\big(1\ominus \theta_j\big)^{y_j}_{p,q^{-1}},
	\end{eqnarray*}
	where $y_j\in\{0,1,\cdots,n\}, \displaystyle\sum_{j=1}^{k}y_j\leq n, s_j=\displaystyle\sum_{i=1}^{j}y_j, 0<\theta_j<1$, and $j\in\{1,2,\cdots,k\}.$
	\item[(d)] The probability function of the negative  deformed multinomial distribution of the second kind with parameters $n, \big(\theta_1,\theta_2,\cdots,\theta_k\big), p$ and $q$ is furnished by:
	\begin{small}
		\begin{eqnarray*}
			P\big(W_1=w_1,\cdots,W_k=w_k\big)=\bigg[\begin{array}{c} n+s_k-1  \\ w_1,w_2,\cdots,w_k\end{array} \bigg]^Q_{p,q}\prod_{j=1}^{k}{\theta^{w_j}_j \big(1\ominus \theta_j\big)^{n+s_k-s_{j}}_{p,q^{-1}}}.
		\end{eqnarray*}
		Furthermore, their recursion relations are given as follows:
		\begin{eqnarray*}
			P_{x+1}= \Big[n-\sum_{j=1}^{k}x_j\Big]^Q_{k,p,q}\prod_{j=1}^{k}\frac{\theta_j\,\big(1\ominus \theta_j\big)_{p,q^{-1}}}{[x_{j}+1]^Q_{p,q}}P_{x},\,\mbox{with}\quad P_0= \prod_{j=1}^{k}\big(1\ominus\theta_j\big)^{n}_{p,q^{-1}},
		\end{eqnarray*}
	\end{small}
	where $w_j\in\mathbb{N}, s_j=\displaystyle\sum_{i=1}^{j}w_j, 0<\theta_j<1$, and $j\in\{1,2,\cdots,k\}.$
\end{itemize}
\end{enumerate}
	
\section{Concluding remarks}
The multinomial coefficient and the  multinomial probability distribution and the negative multinomial probability distribution from the $\mathcal{R}(p,q)-$ deformed quantum algebras have been examined and discussed. Particular cases have been deduced. The numerical interpretation of these probabilities distributions is in preparation. 
\section*{Acknowledgements}
This research was partly supported by the SNF Grant No. IZSEZ0\_206010 .  
	\end{document}